\documentclass[a4paper,onecolumn,11pt,accepted=2024-05-06]{quantumarticle}
\pdfoutput=1
\usepackage[utf8]{inputenc}
\usepackage[english]{babel}
\usepackage[T1]{fontenc}
\usepackage{amsmath}
\usepackage{hyperref}
\usepackage[numbers,sort&compress]{natbib}
\usepackage{tikz}
\usepackage{lipsum}
\usepackage{amsmath,amsfonts,amssymb,caption,color,epsfig,graphics,graphicx,hyperref,latexsym,mathrsfs,revsymb,theorem,url,verbatim,enumerate,epstopdf,tikz,float,multirow,booktabs,appendix,enumitem}

\hypersetup{colorlinks,linkcolor={blue},citecolor={red},urlcolor={blue}}
\usetikzlibrary{arrows, decorations.markings}

\tikzstyle{vecArrow} = [thick, decoration={markings,mark=at position
	1 with {\arrow[semithick]{open triangle 60}}},
double distance=1.4pt, shorten >= 5.5pt,
preaction = {decorate},
postaction = {draw,line width=1.4pt, white,shorten >= 4.5pt}]
\tikzstyle{innerWhite} = [semithick, white,line width=1.4pt, shorten >= 4.5pt]

\newtheorem{definition}{Definition}
\newtheorem{proposition}{Proposition}
\newtheorem{lemma}{Lemma}

\newtheorem{theorem}{Theorem}
\newtheorem{corollary}[definition]{Corollary}
\newtheorem{conjecture}[definition]{Conjecture}

\newtheorem{remark}[definition]{Remark}
\newtheorem{example}{Example}
\newtheorem{question}[definition]{Question}

\def\bcj{\begin{conjecture}}
	\def\ecj{\end{conjecture}}
\def\bcr{\begin{corollary}}
	\def\ecr{\end{corollary}}
\def\bd{\begin{definition}}
	\def\ed{\end{definition}}
\def\bea{\begin{eqnarray}}
\def\eea{\end{eqnarray}}
\def\bem{\begin{enumerate}}
	\def\eem{\end{enumerate}}
\def\bex{\begin{example}}
	\def\eex{\end{example}}
\def\bim{\begin{itemize}}
	\def\eim{\end{itemize}}
\def\bl{\begin{lemma}}
	\def\el{\end{lemma}}
\def\bma{\begin{bmatrix}}
	\def\ema{\end{bmatrix}}
\def\bpf{\begin{proof}}
	\def\epf{\end{proof}}
\def\bpp{\begin{proposition}}
	\def\epp{\end{proposition}}
\def\bqu{\begin{question}}
	\def\equ{\end{question}}
\def\br{\begin{remark}}
	\def\er{\end{remark}}
\def\bt{\begin{theorem}}
	\def\et{\end{theorem}}

\def\squareforqed{\hbox{\rlap{$\sqcap$}$\sqcup$}}
\def\qed{\ifmmode\squareforqed\else{\unskip\nobreak\hfil
		\penalty50\hskip1em\null\nobreak\hfil\squareforqed
		\parfillskip=0pt\finalhyphendemerits=0\endgraf}\fi}
\def\endenv{\ifmmode\;\else{\unskip\nobreak\hfil
		\penalty50\hskip1em\null\nobreak\hfil\;
		\parfillskip=0pt\finalhyphendemerits=0\endgraf}\fi}
\newenvironment{proof}{\noindent \textbf{{Proof.~} }}{\qed}

\def\Dbar{\leavevmode\lower.6ex\hbox to 0pt
	{\hskip-.23ex\accent"16\hss}D}
\makeatletter
\def\url@leostyle{%
	\@ifundefined{selectfont}{\def\UrlFont{\sf}}{\def\UrlFont{\small\ttfamily}}}
\makeatother
\def\bcj{\begin{conjecture}}
	\def\ecj{\end{conjecture}}
\def\bcr{\begin{corollary}}
	\def\ecr{\end{corollary}}
\def\bd{\begin{definition}}
	\def\ed{\end{definition}}
\def\bea{\begin{eqnarray}}
\def\eea{\end{eqnarray}}
\def\bem{\begin{enumerate}}
	\def\eem{\end{enumerate}}
\def\bex{\begin{example}}
	\def\eex{\end{example}}
\def\bim{\begin{itemize}}
	\def\eim{\end{itemize}}
\def\bl{\begin{lemma}}
	\def\el{\end{lemma}}
\def\bpf{\begin{proof}}
	\def\epf{\end{proof}}
\def\bpp{\begin{proposition}}
	\def\epp{\end{proposition}}
\def\bqu{\begin{question}}
	\def\equ{\end{question}}
\def\br{\begin{remark}}
	\def\er{\end{remark}}
\def\bt{\begin{theorem}}
	\def\et{\end{theorem}}

\def\btb{\begin{tabular}}
	\def\etb{\end{tabular}}

\newcommand{\nc}{\newcommand}

\nc{\bbA}{\mathbb{A}} \nc{\bbB}{\mathbb{B}} \nc{\bbC}{\mathbb{C}}
\nc{\bbD}{\mathbb{D}} \nc{\bbE}{\mathbb{E}} \nc{\bbF}{\mathbb{F}}
\nc{\bbG}{\mathbb{G}} \nc{\bbH}{\mathbb{H}} \nc{\bbI}{\mathbb{I}}
\nc{\bbJ}{\mathbb{J}} \nc{\bbK}{\mathbb{K}} \nc{\bbL}{\mathbb{L}}
\nc{\bbM}{\mathbb{M}} \nc{\bbN}{\mathbb{N}} \nc{\bbO}{\mathbb{O}}
\nc{\bbP}{\mathbb{P}} \nc{\bbQ}{\mathbb{Q}} \nc{\bbR}{\mathbb{R}}
\nc{\bbS}{\mathbb{S}} \nc{\bbT}{\mathbb{T}} \nc{\bbU}{\mathbb{U}}
\nc{\bbV}{\mathbb{V}} \nc{\bbW}{\mathbb{W}} \nc{\bbX}{\mathbb{X}}
\nc{\bbZ}{\mathbb{Z}}

\nc{\bA}{{\bf A}} \nc{\bB}{{\bf B}} \nc{\bC}{{\bf C}}
\nc{\bD}{{\bf D}} \nc{\bE}{{\bf E}} \nc{\bF}{{\bf F}}
\nc{\bG}{{\bf G}} \nc{\bH}{{\bf H}} \nc{\bI}{{\bf I}}
\nc{\bJ}{{\bf J}} \nc{\bK}{{\bf K}} \nc{\bL}{{\bf L}}
\nc{\bM}{{\bf M}} \nc{\bN}{{\bf N}} \nc{\bO}{{\bf O}}
\nc{\bP}{{\bf P}} \nc{\bQ}{{\bf Q}} \nc{\bR}{{\bf R}}
\nc{\bS}{{\bf S}} \nc{\bT}{{\bf T}} \nc{\bU}{{\bf U}}
\nc{\bV}{{\bf V}} \nc{\bW}{{\bf W}} \nc{\bX}{{\bf X}}
\nc{\bZ}{{\bf Z}} \nc{\bm}{{\bf m}} \nc{\bv}{{\bf v}}
\nc{\ba}{{\bf a}} \nc{\be}{{\bf e}} \nc{\bu}{{\bf u}}
\nc{\brr}{{\bf r}}

\nc{\cA}{{\cal A}} \nc{\cB}{{\cal B}} \nc{\cC}{{\cal C}}
\nc{\cD}{{\cal D}} \nc{\cE}{{\cal E}} \nc{\cF}{{\cal F}}
\nc{\cG}{{\cal G}} \nc{\cH}{{\cal H}} \nc{\cI}{{\cal I}}
\nc{\cJ}{{\cal J}} \nc{\cK}{{\cal K}} \nc{\cL}{{\cal L}}
\nc{\cM}{{\cal M}} \nc{\cN}{{\cal N}} \nc{\cO}{{\cal O}}
\nc{\cP}{{\cal P}} \nc{\cQ}{{\cal Q}} \nc{\cR}{{\cal R}}
\nc{\cS}{{\cal S}} \nc{\cT}{{\cal T}} \nc{\cU}{{\cal U}}
\nc{\cV}{{\cal V}} \nc{\cW}{{\cal W}} \nc{\cX}{{\cal X}}
\nc{\cZ}{{\cal Z}}

\nc{\hA}{{\hat{A}}} \nc{\hB}{{\hat{B}}} \nc{\hC}{{\hat{C}}}
\nc{\hD}{{\hat{D}}} \nc{\hE}{{\hat{E}}} \nc{\hF}{{\hat{F}}}
\nc{\hG}{{\hat{G}}} \nc{\hH}{{\hat{H}}} \nc{\hI}{{\hat{I}}}
\nc{\hJ}{{\hat{J}}} \nc{\hK}{{\hat{K}}} \nc{\hL}{{\hat{L}}}
\nc{\hM}{{\hat{M}}} \nc{\hN}{{\hat{N}}} \nc{\hO}{{\hat{O}}}
\nc{\hP}{{\hat{P}}} \nc{\hR}{{\hat{R}}} \nc{\hS}{{\hat{S}}}
\nc{\hT}{{\hat{T}}} \nc{\hU}{{\hat{U}}} \nc{\hV}{{\hat{V}}}
\nc{\hW}{{\hat{W}}} \nc{\hX}{{\hat{X}}} \nc{\hZ}{{\hat{Z}}}

\nc{\hn}{{\hat{n}}}

\def\dim{\mathop{\rm Dim}}

\newcommand{\upb}{\mathcal{U}}
\newcommand{\bra}[1]{\langle #1|}
\newcommand{\ket}[1]{| #1\rangle}

\newcommand{\ketbra}[2]{|#1\rangle\!\langle#2|}
\newcommand{\braket}[2]{\langle#1|#2\rangle}

\newcommand{\fl}[2]{\lfloor\frac{#1}{#2}\rfloor}

\def\Dbar{\leavevmode\lower.6ex\hbox to 0pt
	{\hskip-.23ex\accent"16\hss}D}

\newcommand{\indexset}{\Lambda}
\newcommand{\ifel}{\text{ if }}
\newcommand{\0}{{\bf 0}}
\newcommand{\partition}{\cB}
\renewcommand{\mod}{\;\mathrm{mod}\,}

\newcommand{\ops}{\cO}

\begin{document}
	\title{Strong quantum nonlocality and unextendibility without entanglement in $N$-partite systems with odd $N$}
	\author{Yiyun He}
	\affiliation{Department of Mathematics,
		University of California, Irvine, 92697, CA, United States}
	
	\author{Fei Shi}
 \affiliation{QICI Quantum Information and Computation Initiative, Department of Computer Science,
The University of Hong Kong, Pokfulam Road, Hong Kong}
	
	\author{Xiande Zhang}
	\email[]{
 drzhangx@ustc.edu.cn}
	\affiliation{School of Mathematical Sciences,
		University of Science and Technology of China, Hefei, 230026, People's Republic of China}
\affiliation{Hefei National Laboratory, University of Science and Technology of China, Hefei, 230088, China}
	%
	
	\begin{abstract}
     A set of orthogonal product states is strongly nonlocal if it is locally irreducible in every bipartition, which shows the phenomenon of strong quantum nonlocality without entanglement [\href{https://journals.aps.org/prl/abstract/10.1103/PhysRevLett.122.040403}{Phys. Rev. Lett. \textbf{122}, 040403 (2019)}].  Although such a phenomenon has been shown to any three-, four-, and five-partite systems, the existence of strongly nonlocal orthogonal product sets in multipartite systems remains unknown. In this paper, by using a general decomposition of the $N$-dimensional hypercubes, we present strongly nonlocal orthogonal product sets in $N$-partite systems for all odd $N\geq 3$. Based on this decomposition, we give explicit constructions of unextendible product bases in $N$-partite systems for odd $N\geq 3$. Furthermore, we apply our results to quantum secret sharing, uncompletable product bases, and PPT entangled states.
	\end{abstract}		
	\maketitle

\vspace{-0.5cm}
\indent{\textbf{Keywords}}: strong quantum nonlocality, unextendible product bases, hypercubes	

	\section{Introduction}
	
		Quantum nonlocality is one of the most fundamental property in quantum world. Entangled states show Bell nonlocality for
	violating Bell-type inequalities \cite{horodecki2009quantum,brunner2014bell}. However, besides Bell-type nonlocality, there is a different kind of nonlocality which arises from local indistinguishability. A set of orthogonal states is locally indistinguishable if it is impossible to distinguish them under local operations and classical communications (LOCC). Bennett \emph{et al.} showed the phenomenon of quantum nonlocality without entanglement, by presenting a locally indistinguishable orthogonal product basis (OPB) in $\bbC^3\otimes \bbC^3$ \cite{bennett1999quantum}. Later, quantum nonlocality based on local indistinguishability  has received much attention \cite{walgate2000local,ghosh2001distinguishability,Horodecki2003,DivincenzoDavidP2003,de2004distinguishability,GhoshSibasish2004,FanHeng2004,Niset2006,fan2007distinguishing,FengY2009,YuNengkun2012,Bandyopadhyay2012,cosetino2013,Zhang2014,li2015d,zhang2016local,xuguangbao2017,wangyanling2017,zhangzhichao2017,halder2018,xu2021novel,xiong2019positive,zuo2021nonlocal,li2021local,zhu2022nonlocal,Zhen2022,wang2022small,li2023bounds,cao2023locally}.
	
	Recently, a stronger version of local indistinguishability was introduced by Halder \emph{et al.} - local irreducibility  \cite{Halder2019Strong}. A set of orthogonal states is locally irreducible if it is not possible to eliminate one or more states from the set by orthogonality-preserving local measurements. Moreover, a set of orthogonal states is strongly nonlocal if it is locally irreducible in every bipartition. They also showed the phenomenon of strong quantum nonlocality without entanglement, by presenting two strongly nonlocal OPBs in $\bbC^{3}\otimes \bbC^{3}\otimes \bbC^{3}$ and $\bbC^{4}\otimes \bbC^{4}\otimes \bbC^{4}$, respectively. Then strongly nonlocal orthogonal product sets (OPSs) and orthogonal entangled sets (OESs) were also widely investigated \cite{2020Strong,yuan2020strong,Wang2021,Shi2021strongUPB,Shi2022UPB,li2023bounds,zhou2022orthogonal,Shi2022,li2023strongest,xiong2024existence,hu2024strong,bhunia2024strong}. 
 However, the phenomenon of strong quantum nonlocality without entanglement  has been limited to three-, four-, and five-partite systems up to now \cite{shi2021hyper}. It is difficult to show this phenomenon in $N$-partite systems. This is because  the main construction of strongly nonlocal OPSs relies on the decomposition of $N$-dimensional hypercubes \cite{shi2021hyper}, and when $N$ is large, the decomposition can be more complex. In this paper, we will give a general decomposition of the  $N$-dimensional hypercubes for odd $N\geq 3$, and construct strongly nonlocal OPSs in $N$-partite systems for odd $N\geq 3$.

	An unextendible product basis (UPB) is a set of orthonormal product states whose complementary space has	no product states \cite{bennett1999unextendible}. UPBs can be used to construct bound entangled states \cite{bennett1999unextendible,DivincenzoDavidP2003}, Bell-type inequalities without quantum violation \cite{Augusiak2012tight,augusiak2011bell} and
	fermionic systems \cite{Chen2014Unextendible}. UPBs are also connected to quantum nonlocality and strong quantum nonlocality \cite{de2004distinguishability,Shi2022UPB,Shi2021strongUPB}. By using tile structures, Shi \emph{et al.} gave some explicit constructions of UPBs in $\bbC^m\otimes \bbC^n$ \cite{Shi2020Unextendible}. Then by the decomposition of three- and four-dimensional hypercubes, the authors of \cite{Shi2021strongUPB,Agrawal2019Genuinely} showed some UPBs in three- and four-partite systems. In this paper, based on the decomposition  $N$-dimensional hypercubes for odd $N\geq 3$, we give explicit constructions of UPBs in $N$-partite systems for odd $N\geq 3$.


   Here is some brief outlines of some underlying motivations of our work. Strong quantum nonlocality can be used for quantum secret sharing. Suppose that information is encoded into a strongly nonlocal OPS in an $N$-partite system, and sent to $N$ players, 
    where the $N$ players can only communicate classically and perform orthogonality-preserving local measurements. Then then the original information cannot be perfectly  recovered by the $N$ players, even if $k$ ($k<N$) players collude with each other. Thus it is important to construct strongly nonlocal OPSs in $N$-partite systems, which is the first motivation of this work. An OPS is uncompletable if it cannot be extended to a fully OPB. In 2003, DiVincenzo \emph{et al.} proposed an open question, whether there eixsts a UPB which is uncompletable in every bipartition \cite{DivincenzoDavidP2003}.  Recently, Shi \emph{et al.} showed such a UPB exists in abitrary three-, and four-partite system \cite{shi2022unextendible}. 
     However, the existence of such a UPB in arbitrary $N$-partite system remains unknown. This is because  there are few explicit constructions of UPBs in  $N$-partite systems.  We  show that our UPBs will hopefully solve this problem for odd $N\geq 3$.
    This is the second motivation of this work. 
    A mixed state is a PPT state if it is positive under partial transpose (PPT). PPT entangled states corresponds to bound entangled states, where no pure entanglement can be distilled \cite{horodecki1998mixed}. The normalized projector on the orthogonal complement  of the subspace spanned by a UPB  is a PPT entangled state \cite{bennett1999unextendible}.
    If a mixed state is a PPT entangled state across every bipartition, then it  shows that the set of states separable across every bipartition is a proper subset of the set of states PPT across every bipartition
\cite{ranjan2021state}. We will show that our UPBs in  $N$-partite systems with $N=5$ can be used to construct mixed states which are PPT entangled states in every bipartition. This is the third motivation of this work.

    This paper is organized as follows. In Sec. \ref{sec:preliminaries}, we introduce strong quantum nonlocality and UPBs. In Sec. \ref{sec:construction}, we give a decomposition of the hypercube $\bbZ_3^N$, and construct an OPS from this decomposition. We also give a construction of a strongly nonlocal OPS from the decomposition of the hypercube. In Sec. \ref{sec:upb}, we introduce another main result, a construction of UPBs in $(\bbC^3)^{\otimes N}$.  More general results $N$-partite systems for odd $N\geq 3$ can be found in Sec.~\ref{sec:general}. In Sec.~\ref{sec:app}, we give some applications of our results. Finally, we  conclude  in Sec.~\ref{sec:con}.

	\section{Preliminaries}\label{sec:preliminaries}
	In the section, we will introduce the concepts of strong quantum nonlocality and UPBs.
	To simplify the notation, we do not normalize states and operators.
	The positive operator-valued measure (POVM) is a set of positive semi-definite operators $\{E_m=M_m^{\dagger} M_m\}$ acting on $\cH$, a Hilbert space, with a property that $\sum_m{M_m^{\dagger}M_m}=\bbI_{\cH}$, where $\bbI_{\cH}$ is an identity operator on $\cH$. If the set is a POVM, we call each $E_m$ as a POVM element. We focus on this measurement in this paper, and we regard the POVM measurement as trivial measurement if all the POVM elements $E_m$ are proportional to the identity operator.

	In particular, we consider a specific local measurement called orthogonality-preserving local measurement (OPLM). It is performed to distinguish multipartite orthogonal states, and it is defined with the post-measurement states remaining to be orthogonal.

 \begin{definition}\cite{Halder2019Strong}
     A set of multipartite orthogonal states is \emph{locally irreducible} if it is not possible to eliminate one or more states from the set by OPLMs. Furthermore, a set of multipartite orthogonal states is  \emph{strongly nonlocal} if it is locally irreducible for every bipartition of the subsystems.
 \end{definition}

For example, the Bell basis $\{\Psi_i\}_{i=1}^4$ is locally  irreducible, where
\begin{equation*}
\begin{aligned}
    \ket{\Psi_1}=\ket{0}_{A_1}\otimes \ket{0}_{A_2}+\ket{1}_{A_1}\otimes \ket{1}_{A_2},\\
     \ket{\Psi_2}=\ket{0}_{A_1}\otimes \ket{0}_{A_2}-\ket{1}_{A_1}\otimes \ket{1}_{A_2},\\
    \ket{\Psi_3}=\ket{0}_{A_1}\otimes \ket{1}_{A_2}+\ket{1}_{A_1}\otimes \ket{0}_{A_2},\\
    \ket{\Psi_4}=\ket{0}_{A_1}\otimes \ket{1}_{A_2}-\ket{1}_{A_1}\otimes \ket{0}_{A_2}.
    \end{aligned}
\end{equation*}
If $A_1$ performs an OPLM $\{E=M^\dagger M\}$, where $E$ can be written as a $2\times 2$ matrix $(a_{i,j})_{i,j\in \bbZ_2}$ under the basis $\{\ket{0}_{A_1},\ket{1}_{A_1}\}$,  then $\{M\otimes\bbI_{A_2} \ket{\Psi_i}\}_{i=1}^4$ should be mutually orthogonal. Since $\bra{\Psi_1}E\otimes\bbI_{A_2}\ket{\Psi_2}=0$, we obtain $a_{0,0}=a_{1,1}$.  Moreover, since $\bra{\Psi_1}E\otimes\bbI_{A_2}\ket{\Psi_3}=\bra{\Psi_1}E\otimes\bbI_{A_2}\ket{\Psi_4}=0$, we have $a_{0,1}=a_{1,0}=0$. It means that arbitrary OPLM performed by $A_1$ is trvial.   By the symmetry of the Bell basis, arbitrary OPLM performed by $A_2$ is also trvial. 
Thus, the Bell basis $\{\ket{\Psi_i}\}_{i=1}^4$ is locally irreducible.


There is a simple method for showing strong quantum nonlocality.
	\begin{lemma}\cite{Shi2022UPB}\label{lem:cyc}
		Let $\cS:=\{\ket{\psi_j}\}$ be a set of orthogonal states in a multipartite system $\otimes_{i=1}^{N}\cH_{A_i}$. For each $i=1,2,\ldots,N$, define $B_i=\{A_1A_2\ldots A_N\}\setminus \{A_i\}$ be the joint party of all but the $i$th party. Then the set  $\cS$ is strongly nonlocal if the following condition holds for any $1\leq i\leq N$: if party $B_i$  performs any OPLM, then the OPLM is trivial.
	\end{lemma}

Next, we introduce the concept of UPBs.
\begin{definition}\cite{bennett1999unextendible}
    A set of orthogonal product states $\{\ket{\psi_i}\}$ is an unextendible product basis (UPB) if the orthogonal complement of ${\sf Span}\{\ket{\psi_i}\}$    has non-zero dimension and contains no product state. 
\end{definition}

  For example, let
	\begin{equation}
	\begin{aligned}
	&\ket{\psi_1}=\ket{0}_{A_1}\otimes \ket{1}_{A_2}\otimes \ket{+}_{A_3}, \\ 	&\ket{\psi_2}=\ket{1}_{A_1}\otimes \ket{+}_{A_2}\otimes \ket{0}_{A_3}, \\
	&\ket{\psi_3}=\ket{+}_{A_1}\otimes \ket{0}_{A_2}\otimes \ket{1}_{A_3}, \\	&\ket{\psi_4}=\ket{-}_{A_1}\otimes \ket{-}_{A_2}\otimes \ket{-}_{A_3}, \\
	\end{aligned}
	\end{equation}
	where $\ket{\pm}=\ket{0}\pm \ket{1}$. Then
	$\{\ket{\psi_i}\}_{i=1}^4$ is  a UPB in $\bbC^2\otimes \bbC^2 \otimes \bbC^2$.

	Let $(\bbC^{d})^{\otimes N}:=\bbC^d\otimes \bbC^d\otimes \cdots \otimes \bbC^d$, and denote $\bbZ_d^N:=\bbZ_d\times \bbZ_d \times \cdots\times \bbZ_d$,  where $\bbZ_d:=\{0,1,2,\ldots,d-1\}$ and repeat $N$ times.
	We assume that $\{\ket{j}\}_{j\in\bbZ_d}$ is the computational basis of $\bbC^d$, then the computational basis of $(\bbC^{d})^{\otimes N}$ is
	\begin{equation}
	\{\ket{j_1}_{A_1}\otimes \ket{j_2}_{A_2}\otimes\cdots\otimes\ket{j_N}_{A_N}\}_{(j_1,j_2,\ldots,j_N)\in \bbZ_d^N}.
	\end{equation}
    In the space $\bbZ_d^N$, when there is no ambiguity, we denote the vector $(j_1,j_2,\ldots,j_N)$ together with the single set $\{(j_1,j_2,\ldots,j_N)\}$ as
	\begin{equation}
 {\{j_1\}}_{A_1}\times{\{j_2\}}_{A_2}\times\cdots\times{\{j_N\}}_{A_N}.
	\end{equation}
	There are $d^N$ vectors in $\bbZ_d^N$, and all vectors in $\bbZ_d^N$ form an $N$-dimensional hypercube.
	Each product state $\ket{j_1}_{A_1}\otimes \ket{j_2}_{A_2}\otimes\cdots\otimes\ket{j_N}_{A_N}\in (\bbC^{d})^{\otimes N}$ corresponds to a vector $ {\{j_1\}}_{A_1}\times{\{j_2\}}_{A_2}\times\cdots\times{\{j_N\}}_{A_N}\in \bbZ_d^N$ in the hypercube.
	
	For a general product state $\ket{\psi}\in (\bbC^{d})^{\otimes N}$, there exists a unique subset $E_{A_i}\subseteq\bbZ_{d}$ for $1\leq i\leq N$ and nonzero $a_{j_i}^{(A_i)}\in\bbC$ for $j_i\in E_{A_i}$, such that
	
	\begin{equation}
	\ket{\psi}=\left(\sum_{j_1\in E_{A_1}}a^{(A_1)}_{j_1}\ket{j_1}\right)_{A_1}\otimes\left(\sum_{j_2\in E_{A_2}}a^{(A_2)}_{j_2}\ket{j_2}\right)_{A_2}\otimes
	\cdots\otimes\left(\sum_{j_N\in E_{A_N}}a^{(A_N)}_{j_N}\ket{j_N}\right)_{A_N}.
	\end{equation}
    Then the \emph{support set} $E=E_{A_1}\times E_{A_2}\times \cdots\times E_{A_N}\subseteq \bbZ_d^N$ is a  subcube of $\bbZ_d^N$. We denote $|E|:=\prod_{i=1}^N|E_{A_i}|$ as the number of vectors contained in $E$. For instance, in the example above, $\ket{\psi_1} = \ket{0}_{A_1}\otimes \ket{1}_{A_2}\otimes (\ket{0}_{A_3}+\ket{1}_{A_3})$, and its support set is $E = \{0\}_{A_1}\times\{1\}_{A_2}\times \{0,1\}_{A_3}$.
 
    We say a set of subcubes forms a \emph{decomposition} of $\bbZ_d^N$, if any two subcubes are disjoint, and the union of all subcubes is $\bbZ_d^N$.
	
   \section{Construction of OPSs in $(\bbC^3)^{\otimes N}$ for odd $N\geq 3$ }	
   \label{sec:construction}
	
    In this section, we show a decomposition of $\bbZ_3^N$ for odd $N\geq 3$, and we show an OPS from this decomposition.
      When $N=3$, the authors of Ref.~\cite{Agrawal2019Genuinely} gave a decomposition $\{\cB_i\}_{i=1}^9$ of  $\bbZ_3^3$ (we exchange $A_1$ and $A_3$ parties of Ref.~\cite{Agrawal2019Genuinely}), where the 9 subcubes are
	\begin{equation}\label{eq:decomposition_333}
	\begin{aligned}
	\cB_1:=&\{0, 1\}_{A_1}\times\{1, 2\}_{A_2}\times\{2\}_{A_3}, \\
	\cB_2:=&\{0, 1\}_{A_1}\times\{0\}_{A_2}\times\{1,2\}_{A_3},\\
	\cB_3:=&\{0\}_{A_1}\times\{1, 2\}_{A_2}\times\{0, 1\}_{A_3},\\
	\cB_4:=&\{1, 2\}_{A_1}\times\{0, 1\}_{A_2}\times\{0\}_{A_3}, \\
	\cB_5:=&\{1, 2\}_{A_1}\times\{2\}_{A_2}\times\{0, 1\}_{A_3},\\
	\cB_6:=&\{2\}_{A_1}\times\{0, 1\}_{A_2}\times\{1, 2\}_{A_3}, \\
	\cB_7:=&\{0\}_{A_1}\times\{0\}_{A_2}\times\{0\}_{A_3},\\ 
     \cB_8:=&\{1\}_{A_1}\times\{1\}_{A_2}\times\{1\}_{A_3},\\ 
     \cB_9:=&\{2\}_{A_1}\times\{2\}_{A_2}\times\{2\}_{A_3}.
	\end{aligned}
	\end{equation}
	Note that $\cB_i$ contains four vectors of $\bbZ_3^3$ for $1\leq i\leq 6$, and $\cB_i$ contains one vector of $\bbZ_3^3$ for $7\leq i\leq 9$.   See Fig.~\ref{fig:333} for this decomposition.
	\begin{figure}[h]
		\centering
		\includegraphics[scale=0.55]{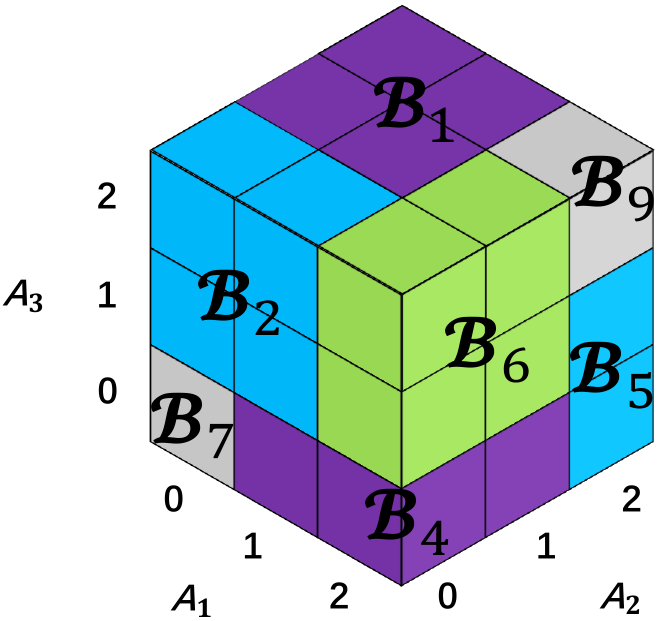}
		\caption{The decomposition for the $\bbZ_3^3$ hypercube.   }\label{fig:333}
	\end{figure}


	Then by these subcubes, the authors constructed an OPB $\cup_{i=1}^9\cB_i$ in $(\bbC^{3})^{\otimes 3}$, where
	\begin{equation}\label{eq:opb333}
	\begin{aligned}
    \ket{\cB_1}&:=\{\ket{\psi_1{(i,j)}}=\ket{\eta_i}_{A_1}\otimes\ket{\xi_j}_{A_2}\otimes\ket{2}_{A_3}\mid (i,j)\in\bbZ_2^2\},\\
	\ket{\cB_2}&:=\{\ket{\psi_2{(i,j)}}=\ket{\eta_i}_{A_1}\otimes\ket{0}_{A_2}\otimes\ket{\xi_j}_{A_3}\mid (i,j)\in\bbZ_2^2\},\\
    \ket{\cB_3}&:=\{\ket{\psi_3{(i,j)}}=\ket{0}_{A_1}\otimes\ket{\xi_i}_{A_2}\otimes\ket{\eta_j}_{A_3}\mid (i,j)\in\bbZ_2^2\},\\
	\ket{\cB_4}&:=\{\ket{\psi_4{(i,j)}}=\ket{\xi_i}_{A_1}\otimes\ket{\eta_j}_{A_2}\otimes\ket{0}_{A_3}\mid (i,j)\in\bbZ_2^2\},\\
	\ket{\cB_5}&:=\{\ket{\psi_5{(i,j)}}=\ket{\xi_i}_{A_1}\otimes\ket{2}_{A_2}\otimes\ket{\eta_j}_{A_3}\mid (i,j)\in\bbZ_2^2\},\\
    \ket{\cB_6}&:=\{\ket{\psi_6   {(i,j)}}=\ket{2}_{A_1}\otimes\ket{\eta_i}_{A_2}\otimes\ket{\xi_j}_{A_3}\mid (i,j)\in\bbZ_2^2\},\\
 \ket{\cB_7}&:=\{\ket{0}_{A_1}\otimes\ket{0}_{A_2}\otimes\ket{0}_{A_3}\},\\
	\ket{\cB_8}&:=\{\ket{1}_{A_1}\otimes\ket{1}_{A_2}\otimes\ket{1}_{A_3}\},\\
	\ket{\cB_9}&:=\{\ket{2}_{A_1}\otimes\ket{2}_{A_2}\otimes\ket{2}_{A_3}\},\\
	\end{aligned}
	\end{equation}
	and $\ket{\eta_i}=\ket{0}+(-1)^{i}\ket{1}$, $\ket{\xi_j}=\ket{1}+(-1)^{j}\ket{2}$, for $i,j\in\bbZ_2$ \cite{Agrawal2019Genuinely}. The authors of Ref.~\cite{shi2021hyper} showed that the OPS $\bigcup_{i=1}^7\ket{\cB_i}\cup\ket{\cB_9}$ is strongly nonlocal, which corresponds to the outermost layer of Fig.~\ref{fig:333}.
	They also gave a similar decomposition for the $5$-dimensional hypercube 
	and showed that the OPS from the outermost layer is strongly nonlocal \cite{shi2021hyper}.

    Before giving the rigorous mathematical construction of the generalized case, one can look back to Figure~\ref{fig:333} -- what makes this decomposition special? There are plenty of different ways to decompose a cube, but not all of them can provide an OPS with nonlocal or UPB property. The crucial part of the construction in Figure~\ref{fig:333} is that: we cannot get a subcube other than $\bbZ_3^n$ by combining some of the subcubes in this partition. We will explain it via the high dimensional construction and the generalized theorems. \\

    Next, we generalize the decomposition of Fig.~\ref{fig:333}  to any $N$-dimensional hypercube  $\bbZ_3^N$ and odd $N\geq 3$.
	\vspace{0.4cm}

	
	\noindent{\bf Construction.} Let $\cB_0=\{1\}_{A_1}\times \{1\}_{A_2}\times\ldots\times \{1\}_{A_N}$, where $N$ is odd.  We denote
	\begin{equation}
	\{\eta\}:=\{0,1\}, \quad \{\xi\}:=\{1,2\}.
	\end{equation}
	Next, assume that  $K\subseteq\{A_1,A_2,\ldots,A_N\}$ and $|K|$ is even. In particular, $|K|=0$ implies that $K=\emptyset$.
	For any fixed $K$, we can construct two subcubes
	\begin{equation}
	\begin{aligned}
	\cC_K:=&\cC_K^{(A_1)}\times \cC_K^{(A_2)} \times \cdots\times \cC_K^{(A_N)},\\
	\cD_K:=&\cD_K^{(A_1)}\times \cD_K^{(A_2)} \times \cdots\times \cD_K^{(A_N)},\\
	\end{aligned}
	\end{equation}
	as follows. Let
	\begin{equation}\label{eq:A1}
	\begin{aligned}
	\cC_{K}^{(A_1)}:= \left\{
	\begin{array}{lll}
	\{0\}_{A_1}  &\text{if} \ A_1\notin K,\\
	\{\eta\}_{A_1}  & \text{if} \ A_1 \in K,
	\end{array}
	\right.
	\\
    \cD_{K}^{(A_1)}:= \left\{
	\begin{array}{lll}
	\{2\}_{A_1}  & \text{if} \ A_1\notin K,\\
	\{\xi\}_{A_1}  & \text{if} \  A_1 \in K.
	\end{array}
	\right.
	\end{aligned}
	\end{equation}
	For $2\leq i \leq N$, $\cC_{K}^{(A_i)}$ and $\cD_{K}^{(A_i)}$ are defined recursively as in Table~\ref{Table:def}.
	\begin{table}[h]
		\centering
		\caption{The construction of $\cC_{K}^{(A_i)}$ and $\cD_{K}^{(A_i)}$ for $0\leq 1 \leq N-1$}\label{Table:def}
		\begin{tabular}{c|c|c}
			$\cP\!\!\in\!\!\{\cC,\cD\}$, if\, & $\cP_{K}^{(A_{i})}\!=\!\{0\}_{\!A_i} \!\text{ or }\! \{\eta\}_{\!A_i}$& $\cP_{K}^{(A_i)}\!=\!\{2\}_{\!A_i} \!\text{ or }\! \{\xi\}_{\!A_i}$\\[3pt]
			\hline
			if {$A_{i+1}\notin K$} & $\cP_{K}^{(A_{i+1})}=\{0\}_{A_{i+1}}$ & $\cP_{K}^{(A_{i+1})}=\{2\}_{A_{i+1}}$\\[5pt]
			\hline
			if {$A_{i+1} \in K$} & $\cP_{K}^{(A_{i+1})}=\{\xi\}_{A_{i+1}}$ & $\cP_{K}^{(A_{i+1})}=\{\eta\}_{A_{i+1}}$\\[5pt]
		\end{tabular}
	\end{table}
	
\noindent	We denote
	\begin{equation}\label{eq:indexset}
	\indexset:=\left\{K\subseteq \{A_1,A_2,\ldots,A_N\}\mid |K| \ \text{is even}\right\}.
	\end{equation}
	Then $|\indexset|=2^{N-1}$. We also denote
	\begin{equation}\label{eq:p_k}
	\partition:     =\{\cB_0\}\cup\bigcup_{\begin{subarray}c
		K\in\indexset\\\cP\in\{\cC,\cD\}
		\end{subarray}}\{\cP_K\}.
	\end{equation}
	where $|\partition|=2^N+1$, and $|\cP_K|=2^{|K|}$.

 To show that the set $\partition$ of $2^{N}+1$ subcubes is a decomposition of $\bbZ_3^N$, we need to show some important properties of this set. These properties are also useful for comprehension and further discussion on the nonlocal OPS and UPB. The proofs of the following two lemmas can be found in Appendix~\ref{appendix:lem-cyclic} and \ref{appendix:lem-decomposition}.
	\begin{lemma}\label{Lemma:cyclic}
		The set $\partition$ of $2^{N}+1$ subcubes is invariant under the cyclic permutation of the $N$
		parties.
	\end{lemma}
	
	
	By the cyclic property, we always assume that
	\begin{equation}\label{eq:cyclic}
  	A_j=A_{(j\mod{N})},  \quad   \text{if} \quad j\geq N+1.
	\end{equation}

	\begin{lemma}\label{lem:decomposition}
		The set $\partition$ of $2^{N}+1$ subcubes is a decomposition of $\bbZ_3^N$. Moreover, each $\cP_K$ contains exactly one vector in $\{0,2\}^N:=\{0,2\}_{A_1}\times \{0,2\}_{A_2}\times\cdots\times \{0,2\}_{A_N}$.
	\end{lemma}
	
	In other words, if we regard the set $\{0,2\}^N$ as the ``vectors at the corners", the lemma above claims that these corner vectors distribute evenly in our construction.
	
	For better understanding, we give the decomposition of $\bbZ_3^3$ in Table~\ref{Table:N3}. We also give the decomposition of $\bbZ_3^5$ in Table~\ref{Table:N5} in Appendix~\ref{appendix:eg}.

	\begin{table}[htbp]
		\centering		
			\caption{The decomposition of $\bbZ_3^3$, where $\{\eta\}=\{0,1\}$, and $\{\xi\}=\{1,2\}$.} \label{Table:N3}
		\begin{tabular}{|c|c|c|}
			\hline
			$K$ & $\cC_K$ & $\cD_K$\\
			\hline
			$\emptyset$ & $\{0\}_{A_1}\times \{0\}_{A_2}\times\{0\}_{A_3}$ & $\{2\}_{A_1}\times \{2\}_{A_2}\times \{2\}_{A_3}$ \\
			\hline
			$\{A_1,A_2\}$ & $\{\eta\}_{A_1}\times \{\xi\}_{A_2}\times\{2\}_{A_3}$ & $\{\xi\}_{A_1}\times \{\eta\}_{A_2}\times \{0\}_{A_3}$ \\
			\hline
			$\{A_1,A_3\}$ & $\{\eta\}_{A_1}\times \{0\}_{A_2}\times\{\xi\}_{A_3}$ & $\{\xi\}_{A_1}\times \{2\}_{A_2}\times \{\eta\}_{A_3}$ \\
			\hline
			$\{A_2,A_3\}$ & $\{0\}_{A_1}\times \{\xi\}_{A_2}\times\{\eta\}_{A_3}$ & $\{2\}_{A_1}\times \{\eta\}_{A_2}\times \{\xi\}_{A_3}$ \\
			\hline
		\end{tabular}
	\end{table}

    \begin{remark}
        We can check the behaviour of Lemma~\ref{Lemma:cyclic} and Lemma~\ref{lem:decomposition} under the case $\bbZ_3^3$. From Table~\ref{Table:N3}, one can see that the construction contains all $\xi\eta0, \eta0\xi, 0\xi\eta$ and   $\eta\xi2,\xi2\eta,2\eta\xi$ in the chart, which satisfies Lemma~\ref{Lemma:cyclic}. And from Figure~\ref{fig:333}, we can see each $\cB_i$ occupies a corner vector in the cube, which explains Lemma~\ref{lem:decomposition}. Roughly speaking, we can think of this abstract construction as a delicate partition, where each subcube in the partition grows from a corner, and the whole pattern is symmetric in a certain sense of rotation.\\
    \end{remark}

    Next we can associate the partition with the tensor product space. For each subcube, we can construct an OPS in $(\bbC^3)^{\otimes N}$. More specifically, for the subcube $\cP_K=\cP_K^{(A_1)}\times \cP_K^{(A_2)}\times\cdots\times \cP_K^{(A_N)} \in\cB$,  we can define
	\begin{equation}\label{eq:ket}
	\ket{\cP_K}_{A_j}:= \left\{
	\begin{array}{lll}
	\ket{0}_{A_j}  &\text{if} \ \cP_K^{(A_j)}=\{0\}_{A_j},\\
	\ket{2}_{A_j}  &\text{if} \ \cP_K^{(A_j)}=\{2\}_{A_j},\\
	\ket{\eta}_{A_j}  &\text{if} \ \cP_K^{(A_{j})}=\{\eta\}_{A_{j}},\\
	\ket{\xi}_{A_{j}}  &\text{if} \ \cP_K^{(A_{j})}=\{\xi\}_{A_{j}},\\
	\end{array}
	\right.
	\end{equation}
	where $\ket{\eta}=\ket{0}\pm\ket{1}$,  $\ket{\xi}=\ket{1}\pm\ket{2}$. Then
	\begin{equation}
	\ket{\cP_K}:=\ket{\cP_K}_{A_1}\otimes \ket{\cP_K}_{A_2}\otimes \cdots\otimes \ket{\cP_K}_{A_N}
	\end{equation}
	is an OPS of size $2^{|K|}$ in $(\bbC^{3})^{\otimes N}$. Further,  we let
	\begin{equation}
	\ket{\cB_0}:=\ket{1}_{A_1}\otimes \ket{1}_{A_2}\otimes\cdots\otimes\ket{1}_{A_N}.
	\end{equation}
    By Lemma~\ref{lem:decomposition}, we can construct an OPB in $(\bbC^3)^{\otimes N}$.

	\begin{lemma}\label{Lem:orthogonal_basis}
		For odd $N$, the set
\begin{equation} 	
	\{\ket{\cB_0}\}\cup\bigcup_{\begin{subarray}c
			K\in\indexset\\\cP\in\{\cC,\cD\}
			\end{subarray}}\\\ket{\cP_K}
\end{equation}
	is an OPB in $(\bbC^3)^{\otimes N}$.
	\end{lemma}

Then we are ready to introduce the first main result in this paper:

	
	\begin{theorem}
		\label{thm:nonlocal}
		In $(\bbC^3)^{\otimes N}$ where $N$ is odd, the OPS
		$\ops := \bigcup_{\begin{subarray}c
			K\in\indexset\\\cP\in\{\cC,\cD\}
			\end{subarray}}\ket{\cP_K}$ is strongly nonlocal, where $|\ops|=3^N-1$.
	\end{theorem}
	
    The proof of the theorem is in Appendix~\ref{appendix:nonlocal}. It mainly use two important lemmas in Ref.~\cite{Shi2022UPB}, which allow us to discuss the POVM matrix directly from the construction of OPS.  Comparing Theorem~\ref{thm:nonlocal} and Lemma~\ref{Lem:orthogonal_basis}, we can deduce that the OPB in Lemma~\ref{Lem:orthogonal_basis} would be a strongly nonlocal OPS after removing the state $\ket{\cB_0}$. More general result can be found in Sec.~\ref{sec:general}.

	\section{UPB in $(\bbC^3)^{\otimes N}$}
	\label{sec:upb}
	
	In this section, we will introduce another main result in this paper. We construct a UPB in $(\bbC^3)^{\otimes N}$ for odd $N$.  First, we consider the simple case $N=3$. Let
	\begin{equation}\label{eq:333S}
	\ket{S}=\left(\sum_{i\in \bbZ_3}\ket{i}\right)_{A_1}\otimes\left(\sum_{j\in \bbZ_3}\ket{j}\right)_{A_2}\otimes\left(\sum_{k\in \bbZ_3}\ket{k}\right)_{A_3}.
	\end{equation}
	The authors of Ref.~\cite{Agrawal2019Genuinely} showed that  $\bigcup_{i=1}^6(\cB_i\setminus\{\ket{\psi_i{(0,0)}}\})\cup\{\ket{S}\}$ given by Eqs.~\eqref{eq:opb333} and \eqref{eq:333S} is a UPB in $(\bbC^{3})^{\otimes 3}$. Inspired by this idea, we next show a UPB in $(\bbC^3)^{\otimes N}$.

    For any subcube $\cP_K=\cP_K^{(A_1)}\times \cP_K^{(A_2)}\times\cdots\times \cP_K^{(A_N)}$,  we define
	\begin{equation}
	\ket{\cP_K^+}_{A_j}:= \left\{
	\begin{array}{lll}
	\ket{0}_{A_j}  &\text{if} \ \cP^{(A_j)}=\{0\}_{A_j},\\
	\ket{2}_{A_j}  &\text{if} \ \cP^{(A_j)}=\{2\}_{A_j},\\
	(\ket{0}+\ket{1})_{A_j}  &\text{if} \ \cP^{(A_{j})}=\{\eta\}_{A_{j}},\\
	(\ket{1}+\ket{2})_{A_j}  &\text{if} \ \cP^{(A_{j})}=\{\xi\}_{A_{j}}\\
	\end{array}
	\right.
	\end{equation}
	and the state
	\begin{equation}
	\ket{\cP_K^+}=\ket{\cP_K^+}_{A_1} \ket{\cP_K^+}_{A_2} \cdots \ket{\cP_K^+}_{A_N}\in \mathrm{Span}\{\ket{v}\mid v\in \cP_K\}.
	\end{equation}
	Also, we define the ``stopper" state
	\begin{equation}
	\ket{S}\!:=\!\left(\sum_{i_1\in \bbZ_3}\ket{i_1}\!\right)_{A_1}\!\!\!\!\!
	\otimes\!\left(\sum_{i_2\in \bbZ_3}\ket{i_2}\!\right)_{A_2}\!\!\!\!\!
	\otimes\cdots\otimes\left(\sum_{i_N\in \bbZ_3}\ket{i_N}\!\right)_{A_N}\!\!\!\!\!\!\!.
	\end{equation}
	Now we can give a UPB in $(\bbC^3)^{\otimes N}$.
	
	\begin{theorem}
		\label{thm:upb}
		In $(\bbC^3)^{\otimes N}$ where $N$ is odd, the set of orthogonal product states
		\begin{equation}
            \label{eq:upb333}
		\upb:=\left\{\ket{S}\right\}\cup \bigcup_{\begin{subarray}c
			K\in\indexset\\\cP\in\{\cC,\cD\}
			\end{subarray}}\left(\ket{\cP_K}\setminus\{\ket{\cP_K^+}\}\right)
		\end{equation} is a UPB of size $3^N-2^N$.
	\end{theorem}

    The proof is given via contradiction. We would like to explain a sketch proof for the space $\bbZ_3^3$, where the partition for this specific example is given in Figure~\ref{fig:333} and Table~\ref{Table:N3}. The proof is arranged as follows:
    \begin{enumerate}
        \item[(i)] If $\ket{\psi}$ is a product state, its support set $E\subseteq \bbZ_3^3$ must be a subcube.
        \item[(ii)] If $\ket{\psi}$ is orthogonal to $\upb$, then $E$ must be a union of some subcubes in Table~\ref{Table:N3}.
        \item[(iii)] It's easy to check from Figure~\ref{fig:333} that it is not possible to satisfy both (i) and (ii) unless $E=\bbZ_3^3$.
        \item[(iv)] Show that $\ket{\psi}$ is proportional to $\ket{S}$, which means that all coefficients in the natural linear expansion of $\ket{\psi}$ have the same value.
    \end{enumerate}
    The detailed proof can be found in Appendix~\ref{appendix:upb}. We note that the general proof for (iii) and (iv) would be delicate because of the complexity of the partition, and we can see from step (iv) that the stopper state $\ket{S}$ is inevitable.
    
    When $d$ is larger than $3$, the construction needs to be generalized. A layer-by-layer generalization of this UPB can be found in Sec.~\ref{sec:general}.

    \section{The construction for general system $\bbC^{d_1}\otimes \bbC^{d_2}\otimes \dots \otimes \bbC^{d_N}$ for odd $N\geq 3$ }
    \label{sec:general}
	We have showed strongly nonlocal OPSs and UPBs in  ${(\bbC^3)}^{\otimes N}$ for odd $N$. In this section, we will give a general result in $\bbC^{d_1}\otimes \bbC^{d_2}\otimes \dots \otimes \bbC^{d_N}$.  First, we need to give a decomposition of the hypercube $\bbZ_{d_1}\times\bbZ_{d_2}\times\dots \times \bbZ_{d_N}$. Then we show a strongly nonlocal OPS in $\bbC^{d_1}\otimes \bbC^{d_2}\otimes \dots \otimes \bbC^{d_N}$. Finally, we construct a UPB in $\bbC^{d_1}\otimes \bbC^{d_2}\otimes \dots \otimes \bbC^{d_N}$. Without loss of generality, we always assume that $3\leq d_1\leq d_2\leq \cdots\leq d_N$, and $N\geq 3$ is odd. Denote $\prod_{i=1}^{N}\{d_i\}_{A_i}:=\{d_1\}_{A_1}\times \{d_2\}_{A_2}\times \cdots \times \{d_N\}_{A_N}$.
	
    \subsection{The decomposition of a general cube}
    We define $\delta = {\lfloor \frac{d_1-1}{2}\rfloor}$ many layers $\cL_1,\cL_2,\dots,\cL_\delta$, where each layer is defined as the difference set
    \begin{equation}
    \begin{aligned}
    \cL_k:=&\prod_{i=1}^N\left\{k-1,k,k+1,\dots,d_i-k-1,d_i-k\right\}_{A_i} \setminus\prod_{i=1}^N \left\{k,k+1,\dots,d_i-k-1\right\}_{A_i}.
    \end{aligned}
    \end{equation}
     We also denote the central block as $\cB_0$, which is defined as
    $$\cB_0:=\prod_{i=1}^N \left\{\delta,\delta+1,\dots,d_i-\delta-1\right\}_{A_i}.$$
    Then it is easy to see that $\bbZ_d^N = \cB_0 \cup \cL_1\cup\dots\cup\cL_\delta$ forms a decomposition of the hypercube $\bbZ_{d_1}\times\bbZ_{d_2}\times\dots \times \bbZ_{d_N}$.

    Next, we will give a further decomposition of each layer $\cL_k$. The decomposition of $\cL_k$ is similar to the decomposition of $\bbZ_3^{N}$, and see Eq.~\eqref{eq:p_k}.  
   For the set $\bigcup_{\begin{subarray}c
		K\in\indexset\\\cP\in\{\cC,\cD\}
		\end{subarray}}\{\cP_K\}$ given in Eq.~\eqref{eq:p_k}, we can obtain a decomposition $\bigcup_{\begin{subarray}c
		K\in\indexset\\\cP\in\{\cC,\cD\}
		\end{subarray}}\{\cP_{k,K}\}$ of $\cL_k$ by the following replacements:
	\begin{itemize}[leftmargin=15px]
	    \item replace $\{0\}_{A_i}$ with $\{k-1\}_{A_i}$;
	    \item replace $\{2\}_{A_i}$ with $\{d_i-k\}_{A_i}$;
	    \item replace $\{\eta\}_{A_i}$ with $\{\eta_k\}_{A_i} \!=\! \{k-1,k,\dots,d_i-k-1\}_{A_i}$;
	    \item replace $\{\xi\}_{A_i}$ with $\{\xi_k\}_{A_i} \!=\! \{k,k+1,\dots,d_i-k\}_{A_i}$.
	\end{itemize}
	
   Therefore, we have the following lemma.
	\begin{lemma}
	    The set $\{\cB_0\}\cup \bigcup_{\begin{subarray}c
	    		K\in\indexset\\
	    		\cP\in\{\cC,\cD\}\\
	    		1\leq k\leq \delta
	    \end{subarray}}\{\cP_{k,K}\}$  is a decomposition of $\bbZ_{d_1}\times\bbZ_{d_2}\times\dots \times \bbZ_{d_N}$.
	\end{lemma}
	
	\subsection{Strongly nonlocal OPS in $\bbC^{d_1}\otimes\bbC^{d_2}\otimes \dots\otimes \bbC^{d_N}$ for odd $N\geq 3$}
 For the subcube $\cP_{k,K}=\cP_{k,K}^{(A_1)}\times \cP_{k,K}^{(A_2)}\times\cdots\times \cP_{k,K}^{(A_N)}$, we define
	\begin{equation}
		\ket{\cP_{k,K}}_{A_j}:= \left\{
		\begin{array}{lll}
		\ket{k-1}_{A_j}  &\text{if} \ \cP^{(A_j)}_{k,K}=\{k-1\}_{A_j},\\
		\ket{d_j-k}_{A_j}  &\text{if} \ \cP^{(A_j)}_{k,K}=\{d_j-k\}_{A_j},\\
		\ket{\eta_k}_{A_j}  &\text{if} \ \cP^{(A_{j})}_{k,K}=\{\eta_k\}_{A_{j}},\\
		\ket{\xi_k}_{A_{j}}  &\text{if} \ \cP^{(A_{j})}_{k,K}=\{\xi_k\}_{A_{j}},\\
		\end{array}
		\right.
	\end{equation}
	where $\ket{\eta_k}_{A_j}$, $\ket{\xi_k}_{A_j}$ are the Fourier basis spanned by  $\left\{\ket{m}_{A_j}\right\}_{m=k-1}^{d_j-k-1}$ and $\left\{\ket{m}_{A_j}\right\}_{m=k}^{d_j-k}$, respectively. That is
	\begin{equation}
	\begin{aligned}
 	\ket{\eta_k}_{A_j}&:=\left\{\sum_{m=k-1}^{d_j-k-1}\kappa_j^{(m-k+1)n}\ket{m}_{A_j}\right\}_{n=0}^{d_j-2k},\\
	\ket{\xi_k}_{A_j}&:=\left\{\sum_{m=k}^{d_j-k}\kappa_j^{(m-k)n}\ket{m}_{A_j}\right\}_{n=0}^{d_j-2k},\\
	\end{aligned}
	\end{equation}
	where $\kappa_j=e^{\frac{2\pi\sqrt{-1}}{d_j-2k+1}}$ is a root of unity.
	 We also denote
	\begin{equation*}
	\ket{\cP_{k,K}}=\ket{\cP_{k,K}}_{A_1}\otimes \ket{\cP_{k,K}}_{A_2}\otimes \cdots \otimes\ket{\cP_{k,K}}_{A_N},
	\end{equation*}
	and
	\begin{equation}
	\ket{\cB_0}:=\ket{\beta_1}_{A_1}\otimes\ket{\beta_2}_{A_2}\otimes\cdots\otimes\ket{\beta_N}_{A_N},
	\end{equation}
	 where $\ket{\beta_j}_{A_j}$ is the Fourier basis spanned by  $\left\{\ket{m}_{A_j}\right\}_{m=\delta}^{d_j-\delta-1}$. That is
	\begin{equation}
 	\ket{\beta_j}_{A_j}:=\left\{\sum_{m=\delta}^{d_j-\delta-1} \tau_j^{(m-\delta)n}\ket{m}_{A_j}\right\}_{n=0}^{d_j-2\delta-1},
	 \end{equation}
where $\tau_j=e^{\frac{2\pi\sqrt{-1}}{d_j-2\delta}}$ is another root of unity.
Now, we have the following lemma.

\begin{lemma}
	When $N$ is odd, in $\bbC^{d_1}\otimes\bbC^{d_2}\otimes \cdots\otimes \bbC^{d_N}$, $3\leq d_1\leq d_2 \leq \cdots\leq d_N$, the set  
    $$\{\ket{\cB_0}\}\cup \bigcup_{\begin{subarray}c
		K\in\indexset\\
		\cP\in\{\cC,\cD\}\\
		1\leq k\leq \delta
		\end{subarray}}\ket{\cP_{k,K}}$$
	is an OPB.
\end{lemma}

We can also give a strongly nonlocal OPS in $\bbC^{d_1}\otimes\bbC^{d_2}\otimes \cdots\otimes \bbC^{d_N}$.
	\begin{theorem}
		\label{thm:nonlocal_general}
		When $N$ is odd, in $\bbC^{d_1}\otimes\bbC^{d_2}\otimes \cdots\otimes \bbC^{d_N}$, the OPS
		$\cO:=\bigcup_{\begin{subarray}c
			K\in\indexset\\\cP\in\{\cC,\cD\}
			\end{subarray}}\ket{\cP_{1,K}}$ is strongly nonlocal, and $$|\cO|=d_1d_2\cdots d_N-(d_1-2)(d_2-2)\cdots(d_N-2).$$
	\end{theorem}

The proof of Theorem~\ref{thm:nonlocal_general} is similar to Theorem~\ref{thm:nonlocal}.		
%

	\subsection{UPB in $\bbC^{d_1}\otimes\bbC^{d_2}\otimes \dots\otimes \bbC^{d_N}$ for odd $N\geq 3$}
We denote
	\begin{equation}
	    \ket{S}\! := \!\left(\sum_{i_1=0}^{d_1-1}\ket{i_1}\!\right)_{\!A_1}\!\!\!\!\!\otimes\left(\sum_{i_2=0}^{d_2-1}\ket{i_2}\!\right)_{\!A_2}\!\!\!\!\!\otimes\cdots\otimes\left(\sum_{i_N=0}^{d_N-1}\ket{i_N}\!\right)_{\!A_N}\!\!\!\!\!.
	\end{equation}
	Also, for each element $\cP_{k,K}$ (and also $\cB_0$) in the decomposition, we denote that
	\begin{equation}
	\ket{\cP_{k,K}^+}_{A_j}:= \left\{
	\begin{array}{lll}
	\ket{k-1}_{A_j}  &\text{if} \ \cP^{(A_j)}=\{k-1\}_{A_j},\\
	\ket{d_j-k}_{A_j}  &\text{if} \ \cP^{(A_j)}=\{d_j-k\}_{A_j},\\
	\sum_{i\in \eta_k}\ket{i}_{A_j}  &\text{if} \ \cP^{(A_{j})}=\{\eta_k\}_{A_{j}},\\
	\sum_{i\in \xi_k}\ket{i}_{A_j}  &\text{if} \ \cP^{(A_{j})}=\{\xi_k\}_{A_{j}}\\
	\end{array}
	\right.
	\end{equation}
	and the product state
	\begin{equation}
	\ket{\cP_{k,K}^+}=\ket{\cP_{k,K}^+}_{A_1}\otimes \ket{\cP_{k,K}^+}_{A_2}\otimes \cdots\otimes \ket{\cP_{k,K}^+}_{A_N}\in \ket{\cP_{k,K}}.
	\end{equation}
	We also denote
	\begin{equation}
	\ket{\cB_0^+}=\left(\sum_{i_1=\delta}^{d_1-\delta-1}\ket{i_1}\right)_{A_1}\otimes\left(\sum_{i_2=\delta}^{d_2-\delta-1}\ket{i_2}\right)_{A_2}\otimes
	\cdots\otimes\left(\sum_{i_N=\delta}^{d_N-\delta-1}\ket{i_N}\right)_{A_N}\in\ket{\cB_0}.
	\end{equation}
	Then we are ready to generalize Theorem~\ref{thm:upb}.

	\begin{theorem}
	\label{thm:upb_general}
		When $N$ is odd, in $\bbC^{d_1}\otimes\bbC^{d_2}\otimes \cdots\otimes \bbC^{d_N}$, $3\leq d_1\leq d_2 \leq \cdots\leq d_N$,  the set of orthogonal product states
		\begin{equation}\label{eq:upb}
		\upb:=\left\{\ket{S}\right\}\cup \left(\ket{\cB_0} \setminus \{\ket{\cB_0^+}\}\right) \cup\bigcup_{\begin{subarray}c
			K\in\indexset\\
			\cP\in\{\cC,\cD\}\\
			1\leq k\leq \delta
			\end{subarray}}\!\!\!\left(\ket{\cP_K}\setminus\{\ket{\cP_K^+}\}\right)
		\end{equation} is a UPB of size $d_1d_2\cdots d_N-2^N\delta$, where $\delta=\fl{d_1-1}{2}$.
	\end{theorem}

        We note that in Theorem~\ref{thm:upb_general}, $\ket{\cB_0} = \{\ket{\cB_0^+}\}$ is a set of a single element when $d_1=d_2=\dots=d_N$ is odd, (which corresponds to the center vector in the cube). And hence there is no such term $\ket{\cB_0}\!\setminus\! \{\ket{\cB_0^+}\}$ in the result in Section~\ref{sec:upb}.
 
        The proof is similar to Theorem~\ref{thm:upb}. A subspace  of $\bbC^{d_1}\otimes\bbC^{d_2}\otimes \cdots\otimes \bbC^{d_N}$ is a  completely entangled subspace if it has no non-zero product state, and the maximum dimension of the completely entangled subspace is $\prod_{i=1}^Nd_i-\sum_{i=1}^Nd_i+N-1$ \cite{parthasarathy2004maximal}. The orthogonal complement of the subspace spanned by a UPB  is a  completely entangled subspace. Then our UPB of size $d_1d_2\cdots d_N-2^N\fl{d_1-1}{2}$ in $\bbC^{d_1}\otimes\bbC^{d_2}\otimes \cdots\otimes \bbC^{d_N}$ ($3\leq d_1\leq d_2 \leq \cdots\leq d_N$, $N$ is odd and $N\geq 3$) can be used to construct a completely entangled subspace with dimension $2^N\fl{d_1-1}{2}$, which is smaller than the maximum dimension.

\section{Applications}\label{sec:app}

In this section, we consider some applications of our results in quantum secret sharing, uncompletable product bases, and PPT entangled states.\\


\noindent\textbf{1. Quantum secret sharing.}

It is known that non-orthogonal states cannot be perfectly distinguished \cite{nielsen2010quantum}. Therefore, to  perfectly distinguish a set of multipartite orthogonal states, it is necessary to use OPLMs. Suppose that information is encoded into a strongly nonlocal orthogonal product set (OPS) in an $N$-partite system and then sent to $N$ players. 
If the $N$ players want to recover the original information,  then they must perform OPLMs. Note that it is impossible to eliminate one or more states from the strongly nonlocal OPS by OPLMs for every bipartition of the subsystems.  If $k$ ($k<N$) players collude with each other, which means they can perform joint OPLMs, then the original information cannot be perfectly  recovered by the $N$ players. To perfectly  recover the original information, the $N$ players can perform a global OPLM.


\noindent\textbf{2. Uncompletable product bases in every bipartition.} 

An OPS in $\bigotimes_{i=1}^N \cH_i$ is an \emph{uncompletable product basis (UCPB)} if it cannot be extended to an OPB in $\bigotimes_{i=1}^N \cH_i$  \cite{DivincenzoDavidP2003}. 
It is known that a UPB must be a UCPB, while the converse is not true in general.
In 2003,  DiVincenzo \emph{et al.} proposed an open question, whether there exists a UPB which is a UCPB for every bipartition of the subsystems. Recently, Shi \emph{et al.} answered this open question by showing such a UPB in arbitrary three-,and four-partite system \cite{shi2022unextendible}. 

Since $\upb$ is invariant under the cyclic permutation of the $N$ parties (Lemma~\ref{Lemma:cyclic}), there is only one case of bipartition when $N=3$, $A_1 | A_2A_3$.  When $N=5$, there are in total three kinds of different bipartitions, $A_1|A_2A_3A_4A_5$, $A_1A_2|A_3A_4A_5$ and $A_1A_3|A_2A_4A_5$. For general $N$, there are $\frac{\sum_{i=1}^{\frac{N-1}{2}}\binom{N}{i}}{N}$ kinds of different bipartitions.
When the number of parties $N$ grows, there would be much more different cases. 

Despite of the complexity of the bipartition types, we believe that our construction $\cU$ gives a UCPB for every bipartition of the subsystems for any odd $N\geq 3$. 
For every bipartition of the subsystems, as $\cU$ is a UPB, it suffices to prove that the orthogonal complement $\cH_{\cU}^{\bot}$ has no basis of orthogonal product states in every biparition. By definition of $\cU$ in Theorem~\ref{eq:upb}, we have $\dim \cH_{\cU}^\bot = 2^N$ and 
   \begin{equation*}
       \cH_{\cU}^\bot = \mathrm{Span}\left(\{\ket{\cB_0^+}\}\cup\{\ket{\cP_K^+}\mid K\in \Lambda, \cP \in\{\cC,\cD\}\}\right) \cap (\mathrm{Span}\ket{S})^\bot.
   \end{equation*}
In other words, $\cH_{\cU}^\bot$ consists of linear combinations of $\ket{\cB_0^+}$ and $\ket{\cP_K^+}$ ($2^N+1$ vectors in total) such that the linear combination is perpendicular to $\ket{S}$. Thus the problem became to prove there is no basis of orthogonal product states in such space in every bipartition, with the help of properties of $\cU$.

We present the conjecture and a possible proof procedure as follows.
\begin{conjecture}
    For any odd $N\geq 3$, the construction $\cU$ defined in Theorem~\ref{thm:upb} is an UCPB for every bipartition of the subsystems.
\end{conjecture}
\noindent {\bf Approach.} For any fixed bipartition of $N$ parties, we propose a possible proof idea with the following steps
\begin{enumerate}[label=(\roman*)]
    \item Suppose $d=3$ and there is a orthogonal product basis $\{\ket{\psi_i}\}_{i=1}^{2^N}$ of $\cH_{\cU}^\bot$ under the bipartition. Then there is at least one $\ket{\psi_j}$ contains $\ket{\cB_0^+}$ in the linear combination of $\ket{\cB_0^+}$ and $\ket{\cP_K^+}$'s with non-zero coefficient. Otherwise, 
    for the complete orthogonal basis $\cU\cup\{\ket{\psi_i}\}_{i=1}^{2^N}$, $\ket{\cB_0^+}$ is perpendicular to every vector of this basis except $\ket{S}$, which implies $\ket{\cB_0^+} \propto \ket{S}$. This is a contradiction.
    \item Consider $\ket{\psi_j}$ containing $\ket{\cB_0^+}$. Similar to the proof of UPB, we might be able to prove the support set of $\ket{\psi_j}$ must be $\bbZ_3^N$.
    \item Because $\ket{\psi_j}$ should be a product state under the certain bipartition, given that the support set of $\ket{\psi_j}$ is $\bbZ_3^N$, we might be able to prove that all the coefficients (under the linear combination in $\ket{\cB_0^+}$ and $\ket{\cP_K^+}$'s) of $\ket{\psi_j}$ are the same. This implies $\ket{\psi_j}\propto \ket{S}$, which contradicts to $\braket{S}{\psi_j}=0$.
    \item Finally, similar to the proof in Section~\ref{sec:general}, we can generalize the result to larger $d$ by recursion.
\end{enumerate}
Although we can apply the approach above to check the result technically with some small party numbers such as $N=5$, it is much challenging to deduce the general result because the large number of bipartition types and different shapes in the cube partition construction. \\

\noindent\textbf{3. PPT entangled states in every biparition.} 

A mixed state is called a PPT state if it is positive under partial transposition (PPT). PPT entangled states corresponds to bound entangled states, where no pure entanglement
can be distilled \cite{horodecki1998mixed}. UPBs can be used to construct PPT entangled states, that is, the normalized projector $\rho_{\cH_{\cU}^{\bot}}$ on the orthogonal complement $\cH_{\cU}^{\bot}$ of the subspace spanned by a UPB $\cU$ is a PPT entangled state  \cite{bennett1999unextendible}. For our normalized UPB $\cU=\{\ket{\psi_i}\}_{i=1}^{m}$ in $\bbC^{d_1}\otimes\bbC^{d_2}\otimes \cdots\otimes \bbC^{d_N}$ with odd $N\geq 3$, where $m=|\cU|=d_1d_2\cdots d_N-2^N\fl{d_1-1}{2}$, the state $\rho_{\cH_{\cU}^{\bot}}$ can be written as
\begin{equation}
    \rho_{\cH_{\cU}^{\bot}}=\frac{1}{2^N\fl{d_1-1}{2}}\left(\bbI-\sum_{i=1}^{m}\ketbra{\psi_i}{\psi_i}\right).
\end{equation}
Next we show that $\rho_{\cH_{\cU}^{\bot}}$ is in fact a PPT entangled state in every bipartition. Let  $S$ be a proper subset of $\{A_1,A_2,\cdots,A_N\}$, then under the partial transpose of the parties $S$,  $\bbI$ is invariant, and $\cU$ is transformed to another OPS $\cU'$. It means that $\rho_{\cH_{\cU'}^{\bot}}$ is a mixed state, and it must be positive. Thus $\rho_{\cH_{\cU}^{\bot}}$ is a PPT state in every bipartition. Furthermore, from the above discussion,  any biproduct state $\ket{\psi}_{S}\otimes \ket{\phi}_{\overline{S}}\in \cH_{\cU}^{\bot}$ cannot span $\cH_{\cU}^{\bot}$ when $N=5$ , where  $\overline{S}=\{A_1,A_2,\cdots,A_N\}\setminus S$, and this implies that   $\rho_{\cH_{\cU}^{\bot}}$ is a PPT entangled state in every biparition. Note that the existence of such states indicate that the set of mixed states which is separable in every bipartition is a proper subset of the set of mixed states which is PPT in every bipartition \cite{ranjan2021state}.

\section{Conclusion and Discussion}\label{sec:con}
In this paper, based on the decomposition of $N$-dimensional hypercubes for odd $N\geq 3$, we constructed strongly nonlocal orthogonal product sets and unextendible product bases in $\bbC^{d_1}\otimes\bbC^{d_2}\otimes \dots\otimes \bbC^{d_N}$ for odd $N\geq 3$,  $d_i\geq 3$, and $1\leq i\leq N$. For applications, we showed that our strongly  nonlocal orthogonal product sets can be used for quantum secret sharing, and our UPBs might be used to construct uncompletable product bases in every biparition,
and PPT entangled states in every biparition.

During the reviewing process of this paper, strongly nonlocal orthogonal product sets in $\bbC^{d_1}\otimes\bbC^{d_2}\otimes \cdots\otimes \bbC^{d_N}$ with even $N$ were found in \cite{zhou2023strong}. Combining our results, the phenomenon of strong quantum nonlocality without entanglement exists in arbitrary $N$-partite system.
There are some interesting problems left. 
The first question is whether one can show that our unextendible product bases are strongly nonlocal for $N\geq 5$? Note that, when $N=3$, strongly nonlocal unextendible product bases has been shown in \cite{Shi2021strongUPB}. What is the minimum size of the strongly nonlocal orthogonal product sets? How to prove our Conjecture 4?

\section*{Acknowledgments}
\label{sec:ack}		
The authors are very grateful to the editor and the
anonymous reviewers for providing many useful suggestions which have greatly improved the presentation
of our paper. 
Y.H. and X.Z. were supported by the Innovation Program for Quantum Science and Technology 2021ZD0302902, the NSFC under Grants No. 12171452 and No. 12231014,  and the National Key Research and Development Programs of China 2023YFA1010200 and 2020YFA0713100.
F.S. acknowledges funding from HKU Seed Fund for Basic
Research for New Staff via Project 2201100596, Guangdong Natural Science Fund via Project 2023A1515012185, National Natural Science Foundation of China (NSFC) via Project No. 12305030 and No. 12347104, Hong Kong Research Grant Council (RGC) via No. 27300823, N\_HKU718/23, and R6010-23, Guangdong Provincial Quantum Science Strategic Initiative GDZX2200001.



	\appendix
	\section{The example of the decomposition of $\bbZ_3^5$}
	\label{appendix:eg}

        In Table~\ref{Table:N5} we explain the decomposition of the hypercube $\bbZ_3^5$ by listing $\cC_K,\cD_K$ with corresponding choices of $K\subseteq\{A_1,A_2,\dots A_N\}$ where $|K|$ is even. The cross product signs are omitted for readability: for example, $\{0\}_{A_1}{}\{0\}_{A_2}{}\{0\}_{A_3}{}\{0\}_{A_4}{}\{0\}_{A_5}$ denotes $\{0\}_{A_1}\times\{0\}_{A_2}\times\{0\}_{A_3}\times\{0\}_{A_4}\times\{0\}_{A_5}$.
        
	\begin{table}[htbp]
		\centering
		\caption{The decomposition of $\bbZ_3^5$, where $\{\eta\}=\{0,1\}$, and $\{\xi\}=\{1,2\}$.} \label{Table:N5}
		\begin{tabular}{|c|c|c|}
			\hline
			$K$ & $\cC_K$ & $\cD_K$\\
			\hline
			$\emptyset$ & $\{0\}_{A_1}{}\{0\}_{A_2}{}\{0\}_{A_3}{}\{0\}_{A_4}{}\{0\}_{A_5}$ & $\{2\}_{A_1}{} \{2\}_{A_2}{}\{2\}_{A_3}{}\{2\}_{A_4}{}\{2\}_{A_5}$ \\
			\hline
			$\{A_1,A_2\}$ & $\{\eta\}_{A_1}{}\{\xi\}_{A_2}{}\{2\}_{A_3}{}\{2\}_{A_4}{}\{2\}_{A_5}$ & $\{\xi\}_{A_1}{} \{\eta\}_{A_2}{}\{0\}_{A_3}{}\{0\}_{A_4}{}\{0\}_{A_5}$ \\
			\hline
			$\{A_1,A_3\}$ & $\{\eta\}_{A_1}{}\{0\}_{A_2}{}\{\xi\}_{A_3}{}\{2\}_{A_4}{}\{2\}_{A_5}$ & $\{\xi\}_{A_1}{} \{2\}_{A_2}{}\{\eta\}_{A_3}{}\{0\}_{A_4}{}\{0\}_{A_5}$ \\
			\hline
			$\{A_1,A_4\}$ & $\{\eta\}_{A_1}{}\{0\}_{A_2}{}\{0\}_{A_3}{}\{\xi\}_{A_4}{}\{2\}_{A_5}$ & $\{\xi\}_{A_1}{} \{2\}_{A_2}{}\{2\}_{A_3}{}\{\eta\}_{A_4}{}\{0\}_{A_5}$ \\
			\hline
			$\{A_1,A_5\}$ & $\{\eta\}_{A_1}{}\{0\}_{A_2}{}\{0\}_{A_3}{}\{0\}_{A_4}{}\{\xi\}_{A_5}$ & $\{\xi\}_{A_1}{} \{2\}_{A_2}{}\{2\}_{A_3}{}\{2\}_{A_4}{}\{\eta\}_{A_5}$ \\
			\hline
			$\{A_2,A_3\}$ & $\{0\}_{A_1}{}\{\xi\}_{A_2}{}\{\eta\}_{A_3}{}\{0\}_{A_4}{}\{0\}_{A_5}$ & $\{2\}_{A_1}{} \{\eta\}_{A_2}{}\{\xi\}_{A_3}{}\{2\}_{A_4}{}\{2\}_{A_5}$ \\
			\hline
			$\{A_2,A_4\}$ & $\{0\}_{A_1}{}\{\xi\}_{A_2}{}\{2\}_{A_3}{}\{\eta\}_{A_4}{}\{0\}_{A_5}$ & $\{2\}_{A_1}{} \{\eta\}_{A_2}{}\{0\}_{A_3}{}\{\xi\}_{A_4}{}\{2\}_{A_5}$ \\
			\hline
			$\{A_2,A_5\}$ & $\{0\}_{A_1}{}\{\xi\}_{A_2}{}\{2\}_{A_3}{}\{2\}_{A_4}{}\{\eta\}_{A_5}$ & $\{2\}_{A_1}{} \{\eta\}_{A_2}{}\{0\}_{A_3}{}\{0\}_{A_4}{}\{\xi\}_{A_5}$ \\
			\hline
			$\{A_3,A_4\}$ & $\{0\}_{A_1}{}\{0\}_{A_2}{}\{\xi\}_{A_3}{}\{\eta\}_{A_4}{}\{0\}_{A_5}$ & $\{2\}_{A_1}{} \{2\}_{A_2}{}\{\eta\}_{A_3}{}\{\xi\}_{A_4}{}\{2\}_{A_5}$ \\
			\hline
			$\{A_3,A_5\}$ & $\{0\}_{A_1}{}\{0\}_{A_2}{}\{\xi\}_{A_3}{}\{2\}_{A_4}{}\{\eta\}_{A_5}$ & $\{2\}_{A_1}{} \{2\}_{A_2}{}\{\eta\}_{A_3}{}\{0\}_{A_4}{}\{\xi\}_{A_5}$ \\
			\hline$
			\{A_4,A_5\}$ & $\{0\}_{A_1}{}\{0\}_{A_2}{}\{0\}_{A_3}{}\{\xi\}_{A_4}{}\{\eta\}_{A_5}$ & $\{2\}_{A_1}{} \{2\}_{A_2}{}\{2\}_{A_3}{}\{\eta\}_{A_4}{}\{\xi\}_{A_5}$ \\
			\hline
			$\{A_1,A_2,A_3,A_4\}$ & $\{\eta\}_{A_1}{}\{\xi\}_{A_2}{}\{\eta\}_{A_3}{}\{\xi\}_{A_4}{}\{2\}_{A_5}$ & $\{\xi\}_{A_1}{} \{\eta\}_{A_2}{}\{\xi\}_{A_3}{}\{\eta\}_{A_4}{}\{0\}_{A_5}$ \\
			\hline
			$\{A_1,A_2,A_3,A_5\}$ & $\{\eta\}_{A_1}{}\{\xi\}_{A_2}{}\{\eta\}_{A_3}{}\{0\}_{A_4}{}\{\xi\}_{A_5}$ & $\{\xi\}_{A_1}{} \{\eta\}_{A_2}{}\{\xi\}_{A_3}{}\{2\}_{A_4}{}\{\eta\}_{A_5}$ \\
			\hline
			$\{A_1,A_2,A_4,A_5\}$ & $\{\eta\}_{A_1}{}\{\xi\}_{A_2}{}\{2\}_{A_3}{}\{\eta\}_{A_4}{}\{\xi\}_{A_5}$ & $\{\xi\}_{A_1}{} \{\eta\}_{A_2}{}\{0\}_{A_3}{}\{\xi\}_{A_4}{}\{\eta\}_{A_5}$ \\
			\hline
			$\{A_1,A_3,A_4,A_5\}$ & $\{\eta\}_{A_1}{}\{0\}_{A_2}{}\{\xi\}_{A_3}{}\{\eta\}_{A_4}{}\{\xi\}_{A_5}$ & $\{\xi\}_{A_1}{} \{2\}_{A_2}{}\{\eta\}_{A_3}{}\{\xi\}_{A_4}{}\{\eta\}_{A_5}$ \\
			\hline
			$\{A_2,A_3,A_4,A_5\}$ & $\{0\}_{A_1}{}\{\xi\}_{A_2}{}\{\eta\}_{A_3}{}\{\xi\}_{A_4}{}\{\eta\}_{A_5}$ & $\{2\}_{A_1}{} \{\eta\}_{A_2}{}\{\xi\}_{A_3}{}\{\eta\}_{A_4}{}\{\xi\}_{A_5}$ \\
			\hline
		\end{tabular}
	\end{table}	
	
	\section{The proof of Lemma~\ref{Lemma:cyclic}}
	\label{appendix:lem-cyclic}
	\begin{proof}
		If $K=\emptyset$, we can obtain that
		\begin{equation}\label{eq:1}
		\begin{aligned}
		\cC_{\emptyset}&=\{0\}_{A_1}\times \{0\}_{A_2}\times\ldots\times \{0\}_{A_N},\\
		\cD_{\emptyset}&=\{2\}_{A_1}\times \{2\}_{A_2}\times\ldots\times \{2\}_{A_N}.
		\end{aligned}
		\end{equation}
		If $K\neq \emptyset$, we fix $K=\{A_{k_1},A_{k_2},\ldots,A_{k_{2i}}\}$, where $k_1<k_2<\ldots<k_{2i}$. By Eq.~\eqref{eq:A1} and Table~\ref{Table:def}, we know that $\cP_K^{(A_{k_j})}\times \cP_K^{(A_{k_{j}+1})}=\{\eta\}_{A_{k_{j}}}\times \{\xi\}_{A_{k_{j}+1}}$ or  $\cP_K^{(A_{k_j})}\times \cP_K^{(A_{k_{j}+1})}=\{\xi\}_{A_{k_j}}\times \{\eta\}_{A_{k_{j}+1}}$ for $1\leq j\leq 2i-1$. Thus, $\{\eta\}$ and $\{\xi\}$ alternatively appear in $\cP_K^{(A_{k_1})}\times \cP_K^{(A_{k_2})}\times\cdots\times \cP_K^{(A_{k_{2i}})}$.

		There are four cases.
		
		1. If $A_1\notin  K$, $A_{N}\notin K$, by Eq.~\eqref{eq:A1} and Table~\ref{Table:def}, we have
		\begin{equation}\label{eq:2}
		\begin{aligned}
		\cC_K&=\{0\}_{A_1}\times\cdots \times \{0\}_{A_{k_1-1}}\times \{\xi\}_{A_{k_1}} \times \cdots \times \{\eta\}_{A_{k_{2i}}} \times \{0\}_{A_{k_{2i}+1}}\times\cdots\times \{0\}_{A_N},\\
		\cD_K&=\{2\}_{A_1}\times\cdots \times \{2\}_{A_{k_1-1}}\times \{\eta\}_{A_{k_1}} \times \cdots \times \{\xi\}_{A_{k_{2i}}} \times \{2\}_{A_{k_{2i}+1}}\times\cdots\times \{2\}_{A_N}.\\
		\end{aligned}
		\end{equation}
		
		2. If $A_1\in  K$, $A_{N}\notin K$, by Eq.~\eqref{eq:A1} and Table~\ref{Table:def}, we have
		\begin{equation}\label{eq:3}
		\begin{aligned}
		\cC_K&=\{\eta\}_{A_1}\times\cdots \times  \{\xi\}_{A_{k_{2i}}} \times \{2\}_{A_{k_{2i}+1}}\times\cdots\times \{2\}_{A_N},\\
		\cD_K&=\{\xi\}_{A_1}\times\cdots \times  \{\eta\}_{A_{k_{2i}}} \times \{0\}_{A_{k_{2i}+1}}\times\cdots\times \{0\}_{A_N}.\\
		\end{aligned}
		\end{equation}
		
		3. If $A_1\notin  K$, $A_{N}\in K$, by Eq.~\eqref{eq:A1} and Table~\ref{Table:def}, we have
		\begin{equation}\label{eq:4}
		\begin{aligned}
		\cC_K&=\{0\}_{A_1}\times\cdots \times \{0\}_{A_{k_1-1}}\times \{\xi\}_{A_{k_1}} \times \cdots \times \{\eta\}_{A_N},\\
		\cD_K&=\{2\}_{A_1}\times\cdots \times \{2\}_{A_{k_1-1}}\times \{\eta\}_{A_{k_1}} \times \cdots \times \{\xi\}_{A_N}.\\
		\end{aligned}
		\end{equation}
		
		\begin{figure}[h]
			\centering
			\begin{tikzpicture}[scale=1]
			\draw (11,0) node (p1) [label=right:$\cP_K^{(A_2)}$] {}
			arc (0:60:1.5cm) node (p2) [label=above:$\cP_K^{(A_1)}$] {}
			arc (60:120:1.5cm) node (p3) [label=above:$\cP_K^{(A_N)}$] {}
			arc (120:180:1.5cm) node (p4) [label=left:$\cP_K^{(A_{N-1})}$] {};
			\draw (p4)
			[dotted, line width=1]arc (180:300:1.5cm) node (p5)  [label=right:$\cP_K^{(A_3)}$] {};
			\draw (p5)
			arc (300:360:1.5cm) node (p6) [label=right:] {};
			
			\filldraw (p1) circle (0.08);
			\filldraw (p2) circle (0.08);
			\filldraw (p3) circle (0.08);
			\filldraw (p4) circle (0.08);
			\filldraw (p5) circle (0.08);
			\end{tikzpicture}
			\caption{The cyclic property of  $\cP_K$  for $\cP\in\{\cC,\cD\}$, where $\cP_K^{(A_{i+1})}$ can be determined by $\cP_K^{(A_i)}$ through Table~\ref{Table:def} for $1 \leq i\leq N-1$, and $\cP_K^{(A_1)}$ can be determined by $\cP_K^{(A_N)}$ through Table~\ref{Table:def}.   } \label{fig:cycle}
		\end{figure}
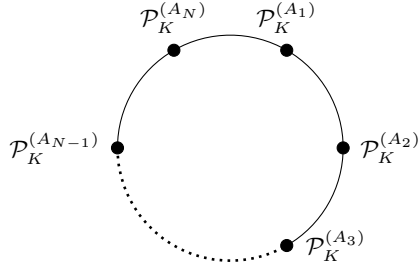

		4. If $A_1\in  K$, $A_{N}\in K$, by Eq.~\eqref{eq:A1}, and Table~\ref{Table:def}, we have
		\begin{equation}\label{eq:5}
		\begin{aligned}
		\cC_K&=\{\eta\}_{A_1}\times\cdots \times  \{\xi\}_{A_N} \\
		\cD_K&=\{\xi\}_{A_1}\times\cdots \times  \{\eta\}_{A_N}.
		\end{aligned}
		\end{equation}
		
		By the four cases and Eq.~\eqref{eq:1}, we know that $\cP_K^{(A_1)}$ can also be determined by $\cP_K^{(A_N)}$ through Table~\ref{Table:def} for $\cP\in\{\cC,\cD\}$. It means that $\cP_K$ forms a cycle for $\cP\in\{\cC,\cD\}$, where $\cP_K^{(A_{i+1})}$ can be determined by $\cP_K^{(A_i)}$ through Table~\ref{Table:def} for $1 \leq i\leq N-1$, and $\cP_K^{(A_1)}$ can be determined by $\cP_K^{(A_N)}$ through Table~\ref{Table:def}.   See Fig.~\ref{fig:cycle} for this phenomenon.
		
	For any $\cP_K = \cP_{K}^{(A_1)}\times \cP_{K}^{(A_2)}\times\dots\times \cP_{K}^{(A_N)}$, it becomes
    $$\cP_{K'}' = \cP_{K}^{(A_i)}\times \cP_{K}^{(A_{i+1})}\times\dots\times \cP_{K}^{(A_N)}\times \cP_{K}^{(A_1)}\times \dots\times \cP_{K}^{(A_{i-1})}$$
    after a cyclic permutation of the parties $\{A_1,A_2,\ldots, A_N\}$. Then to prove $\partition$ is invariant under cyclic permutations, it suffices to prove that $\cP_{K'}'\in \partition$.
    We can define $K'$ as the image of $K$ after the permutation.
    By the cyclic property of  $\cP_K$ in Fig.~\ref{fig:cycle},  the components of $\cP'_{K'}$ satisfies Table \ref{Table:def}. Then  we have $\cP'_{K'}\in \partition$, i.e., $\cP'_{K'}=\cC_{K'}$ if $\cP_{K}^{(A_i)}=\{0\}_{A_i}$ or $\{\eta\}_{A_i}$; $\cP'_{K'}=\cD_{K'}$ if $\cP_{K}^{(A_i)}=\{2\}_{A_i}$ or  $\{\xi\}_{A_i}$. Thus the set $\partition$ of $2^{N}+1$ subcubes  is invariant under the cyclic permutation of the $N$ parties.
	\end{proof}
	
	\section{The proof of Lemma~\ref{lem:decomposition}}
	\label{appendix:lem-decomposition}
	\begin{proof}
		First, we need to show that $\partition$ contains $3^N$ vectors. 	For each subcube $\cP_K = \cP_{K}^{(A_1)}\times \cP_{K}^{(A_2)}\times\dots\times \cP_{K}^{(A_N)}$, it contains $2^{|K|}$ vectors. Then $\partition$  contains
			\begin{equation}
			\label{combinatorics}
		1+\sum_{i=0}^{\frac{N-1}{2}}2\binom{N}{2i}2^{2i}=3^{N}
			\end{equation}
	vectors. Next, for any vector $\{j_1\}_{A_1}\times \{j_2\}_{A_2}\times\cdots\times \{j_N\}_{A_N}\in \bbZ_3^N$, we need to find a $\cP_K\in\partition$ such that $\{j_1\}_{A_1}\times \{j_2\}_{A_2}\times\cdots\times \{j_N\}_{A_N}\in\cP_K$.
		
		There are two possible cases.
		
		1. If $j_k=1$ for all $1\leq k\leq N$, then $\{j_1\}_{A_1}\times \{j_2\}_{A_2}\times \cdots \times \{j_N\}_{A_N}=\cB_0$.
		
		2. If there exists a $j_k\neq 1$ for $1\leq k\leq N$,  then $j_k\in\{0,2\}$. Without loss of generality, we assume that $j_k=0$. Then we can find $\cP_K^{(A_{k})}={\{0\}}_{A_k}$ or ${\{\eta\}}_{A_k}$  for $\cP\in \{\cC,\cD\}$. Next,  if $j_{k+1}=0$, then by Table~\ref{Table:def}, we can determine that $\cP_K^{(A_{k+1})}={\{0\}}_{A_{k+1}}$ and $A_{k+1}\notin K$; if  $j_{k+1}\neq 0$, then by Table~\ref{Table:def}, we can determined that $\cP_K^{(A_{k+1})}=\{\xi\}_{A_{k+1}}$ and $A_{k+1}\in K$. By Fig.~\ref{fig:cycle}, we can repeat this process $N$ times through Table~\ref{Table:def}. Then we can determine $\cP_K^{(A_{i})}$ for $1\leq i\leq N$, and determine whether $A_{i}\in K$. Furthermore, $\cP=\cC$ if  $\cP_K^{(A_{1})}= \{0\}_{A_1}$ or $\{\eta\}_{A_1}$, and $\cP=\cD$ if $\cP_K^{(A_{1})}=\{2\}_{A_1}$ or $\{\xi\}_{A_1}$. In this way, we can find the subcube $\cP_K$ such that  $\{j_1\}_{A_1}\times \{j_2\}_{A_2}\times\cdots\times \{j_N\}_{A_N}\in\cP_K$.
		
		Thus the $2^N+1$ subcubes of  $\partition$ is a decomposition of $\bbZ_3^N$. From the above proof, we know that each $\cP_K$ contains exactly one vector in $\{0,2\}^N$.
	\end{proof}
	
	\section{The proof of Theorem~\ref{thm:nonlocal}}
	\label{appendix:nonlocal}
	Before proving the theorem, we need to introduce some notations and two useful lemmas given in \cite{Shi2022UPB}.  Given an $n$-dimensional Hilbert space, say $\cH_n$, we can take an orthonormal basis $\{|0\rangle,|1\rangle,\cdots, |n-1\rangle\}$. Then consider an operator $E$ on $\cH_n$. It will have a matrix representation under the basis, and we also denote it as $E$. When there is no ambiguity, we will not distinguish the notation of the operator and the matrix.
	
	Next, for the $n\times n$ matrix $E:=\sum_{i=0}^{n-1}\sum_{j=0}^{n-1} a_{i,j}|i\rangle\langle j|$, we introduce the following notation
	\begin{equation*}
	{}_\mathcal{S}E_{\mathcal{T}}:=\sum_{|s\rangle \in \mathcal{S}}\sum_{|t\rangle \in \mathcal{T}}a_{s,t} |s\rangle\langle t|, \quad  \text{ where }\mathcal{S},\mathcal{T}\subseteq \{|0\rangle,|1\rangle,\cdots, |n-1\rangle\}.
	\end{equation*}
	In other words, ${}_\mathcal{S}E_{\mathcal{T}}$ is a sub-matrix of $E$, with  $\mathcal{S},\mathcal{T}$ indicating the chosen row coordinates and column coordinates. In particular, when $\cS=\cT$, we simplify the notation as $E_{\cS}:={}_{\cS}E_{\cS}$. Moreover, we say that an orthogonal set $\{\ket{\psi_i}\}_{i\in\bbZ_s}$ is \emph{spanned} by $\cS\subseteq\{|0\rangle,|1\rangle,\cdots, |n-1\rangle\}$, if any $\ket{\psi_i}$ can be linearly generated by states from $\cS$.
	
	\begin{lemma}[{\bf Block Zeros Lemma} \cite{Shi2022UPB}]
		\label{Lemma:1}
		Let  an  $n\times n$ matrix $E=(a_{i,j})_{i,j\in\bbZ_n}$ be the matrix representation of an operator  $E=M^{\dagger}M$  under the basis  $\cB:=\{\ket{0},\ket{1},\ldots,\ket{n-1}\}$. Given two nonempty disjoint subsets $\cS$ and $\cT$ of $\cB$, assume  that  $\{\ket{\psi_i}\}_{i=0}^{s-1}$, $\{\ket{\phi_j}\}_{j=0}^{t-1}$ are two orthogonal sets  spanned by $\cS$ and $\cT$ respectively, where $s=|\cS|,$ and $t=|\cT|.$  If  $\langle \psi_i| E| \phi_j\rangle =0$
for any $i\in \mathbb{Z}_s,j\in\mathbb{Z}_t$(we call these zero conditions), then   ${}_\mathcal{S}E_{\mathcal{T}}=\mathbf{0}$  and  ${}_\mathcal{T}E_{\mathcal{S}}=\mathbf{0}$.
	\end{lemma}
	\begin{lemma}[{\bf Block Trivial Lemma} \cite{Shi2022UPB}]
		\label{Lemma:2}
		Let  an  $n\times n$ matrix $E=(a_{i,j})_{i,j\in\bbZ_n}$ be the matrix representation of an operator  $E=M^{\dagger}M$  under the basis  $\cB:=\{\ket{0},\ket{1},\ldots,\ket{n-1}\}$. Given a nonempty  subset $\cS:=\{\ket{u_0},\ket{u_1},\ldots,\ket{u_{s-1}}\}$  of $\cB$, let $\{\ket{\psi_j} \}_{j=0}^{s-1}$ be an orthogonal  set spanned by $\cS$.     Assume that $\langle \psi_i|E |\psi_j\rangle=0$ for any $i\neq j\in \mathbb{Z}_s$.  If there exists a state $|u_t\rangle \in\cS$,  such that $ {}_{\{|u_t\rangle\}}E_{\cS\setminus \{|u_t\rangle\}}=\mathbf{0}$   and $\langle u_t|\psi_j\rangle \neq 0$  for any $j\in \mathbb{Z}_s$,   then  $E_{\cS}\propto \mathbb{I}_{\cS}$. (Note that if we consider $\{\ket{\psi_j} \}_{j=0}^{s-1}$ as the Fourier basis, i.e. $\ket{\psi_j}=\sum_{i=0}^{s-1}w_s^{ij}\ket{u_i}$ for  $j\in \mathbb{Z}_s$, then it must have $\langle u_t|\psi_j\rangle \neq 0$  for any $j\in \mathbb{Z}_s$).
	\end{lemma}
	
	\begin{proof}
		By Lemmas~\ref{lem:cyc} and \ref{Lemma:cyclic}, we only need to show that $A_2,\dots, A_N$ can only perform a trivial OPLM. Assume that $A_2,\dots, A_N$ come together to perform an OPLM $\{E=M^\dagger M\}$, where $E=(a_{i_2i_3\ldots i_{N},j_2j_3\ldots j_{N}})_{\{i_2\}_{A_2}\times\cdots\times \{i_N\}_{A_N},\{j_2\}_{A_2}\times\cdots\times \{j_N\}_{A_N}\in\bbZ_3^{N-1}}$,  then $\{\bbI_{A_1}\otimes M\ket{\psi}: \ket{\psi}\in \cO\}$ should be mutually orthogonal. Note that if $a_{i_2i_3\ldots i_{N},j_2j_3\ldots j_{N}}=0$, then $a_{j_2j_3\ldots j_{N},i_2i_3\ldots i_{N}}=0$.
		
		 For any vector $v=\{j_1\}_{A_1}\times \{j_2\}_{A_2}\times \cdots \times   \{j_N\}_{A_N}\in\bbZ_3^N$, we denote $v^{(A_i)}=\{j_i\}_{A_i}$ for $1\leq i\leq N$.
		  Let $A_J := \{A_2,A_3,\dots,A_N\}$. For any subcube $\cP_K=\cP_K^{(A_1)}\times \cP_K^{(A_2)}\times\cdots\times \cP_K^{(A_N)}$, we denote $\cP_K^{(A_J)}:=\cP_K^{(A_2)}\times\cdots\times \cP_K^{(A_N)}$, and $\ket{\cP_K^{(A_J)}}:=\ket{\cP_K}_{A_2}\otimes\cdots\otimes \ket{\cP_K}_{A_N}$.  For two different subcubes $\cP_K, \cP'_{K'}$, denote $_{\cP_K^{(A_J)}} E_{{\cP'_{K'}}^{(A_J)}}:=(a_{u,v})_{u\in \cP_K^{(A_J)}, v\in {\cP'_{K'}}^{(A_J)}}$.
	 	For any state $\ket{\psi} = \ket{\psi}_{A_1}  \ket{\psi}_{A_2}\dots\ket{\psi}_{A_N} \in\cO$, we also denote
	$\ket{\psi^{(A_J)}}= \ket{\psi}_{A_2}\dots\ket{\psi}_{A_N}$.

		Then for two different $\ket{\psi},\ket{\phi}\in \ops$, we have
		\begin{equation}
		0 =\bra{\psi}\bbI_{A_1}\otimes E\ket{\phi}=  \braket{\psi}{\phi}_{A_1} \bra{\psi^{(A_J)}}E\ket{\phi^{(A_J)}}.
		\end{equation}
		Therefore, if $\braket{\psi}{\phi}_{A_1}\neq 0$, then  $\bra{\psi^{(A_J)}}E\ket{\phi^{(A_J)}}=0$. We need to use these orthogonality relations to show that $E\propto \bbI$. \\
		
\noindent		{\bf Step 1:} Let $\cP_K,\cP'_{K'}$ be any two different subcubes with $\cP_K^{(A_1)}\cap{\cP'_{K'}}^{(A_1)}\neq \emptyset$. Then by Eq.~\eqref{eq:ket}, we can always find a state $\ket{\psi}_{A_1}\in \ket{\cP_K}_{A_1}$ and a state $\ket{\phi}_{A_1}\in \ket{\cP'_{K'}}_{A_1}$, such that $\braket{\psi}{{\phi}}_{A_1}\neq 0$. Applying Lemma~\ref{Lemma:1} to $\ket{\cP_K^{(A_J)}}$ and $\ket{{\cP'_{K'}}^{(A_J)}}$, we obtain
		\begin{equation}\label{eq:zero}
		_{\cP_K^{(A_J)}}E_{{\cP'_{K'}}^{(A_J)}} = \0.
		\end{equation}
For any $\cP_K$, assume $v\in \cP_K^{(A_J)}$. Then for any $u\notin  \cP_K^{(A_J)}\subset\bbZ_3^{N-1}$, we can find a $\cP'_{K'}$ such that $u\in {\cP'_{K'}}^{(A_J)}$ and ${\cP'_{K'}}^{(A_1)}\cap \cP_K^{(A_1)}=\emptyset$. By Eq.~\eqref{eq:zero}, we obtain that $a_{v,u}=0$.

\noindent	{\bf Step 2:} Define two vectors $v_0=\{0\}_{A_2}\times\{0\}_{A_3}\times\cdots\times\{0\}_{A_N}\in\bbZ_3^{N-1}$ and $v_1 = \{1\}_{A_2}\times\{0\}_{A_3}\times\cdots\times\{0\}_{A_N}\in\bbZ_3^{N-1}$. We can choose two subcubes
	\begin{eqnarray}
\cC_{\emptyset} &=& \{0\}_{A_1}\times \{0\}_{A_2}\times \{0\}_{A_3}\times \dots\times\{0\}_{A_N},\\
\cD_{\{A_1,A_2\}}& =& \{\xi\}_{A_1}\times \{\eta\}_{A_2}\times\{0\}_{A_3}\times\{0\}_{A_4}\times \dots\times\{0\}_{A_N},
\end{eqnarray}
such that $v_0=\cC_{\emptyset}^{(A_J)}$, and $v_1\in\cD_{\{A_1,A_2\}}^{(A_J)}$.
By Step 1, we obtain that $a_{v_0,v}=0$ for any $v\neq v_0\in\bbZ_3^{N-1}$, and   $a_{v_1,u}=0$ for any $u\neq v_0,v_1\in\bbZ_3^{N-1}$. Then  applying Lemma~\ref{Lemma:2} to $\ket{\cD_{\{A_1,A_2\}}^{(A_J)}}$, we have $a_{v_0,v_0} = a_{v_1,v_1} = k$.  Define
		\begin{equation}
		M = \left\{v\in\bbZ_3^{N-1}\mid a_{v,v}=k, a_{v,w}=0 \; \text{for any} \; w\neq v\right\}.
		\end{equation}
Then we have shown that $v_0,v_1\in M$. In order to show $E=k\bbI$, we only need to prove that  $M=\bbZ_3^{N-1}$.
		

\noindent	{\bf Step 3:} Assume  $v\in M$, then there must exist a $\cP_K$ such that $v\in \cP_K^{(A_J)}$. For any $u\notin \cP_K^{(A_J)}$, we must have $a_{v',u}=0$ for any $v'\in \cP_K^{(A_J)}$ by Step 1. Applying Lemma~\ref{Lemma:2} to $\ket{\cP_K^{(A_J)}}$, we have
\begin{equation}\label{eq:M}
 \cP_K^{(A_J)}\subseteq M.
\end{equation}
By using this fact, we have
\begin{equation}\label{eq:1111}
\begin{aligned}
v_1\in M &\Longrightarrow \{\xi\}_{A_2}\times\{\eta\}_{A_3}\times\{0\}_{A_4}\times \dots\times\{0\}_{A_N} \subseteq M,\\
&\Longrightarrow \{1\}_{A_2}\times \{1\}_{A_3}\times \{0\}_{A_4}\times \cdots \times \{0\}_{A_N} \in M,\\
&\Longrightarrow \{\eta\}_{A_2}\times\{\xi\}_{A_3}\times\{\eta\}_{A_4}\times\{0\}_{A_5}\times \dots\times\{0\}_{A_N} \subseteq M,\\
&\Longrightarrow \{1\}_{A_2}\times \{1\}_{A_3}\times \{1\}_{A_4}\times \{0\}_{A_5}\cdots \{0\}_{A_N}\in M,\\
&\Longrightarrow \dots,\\
&\Longrightarrow \{1\}_{A_2}\times \{1\}_{A_3}\times \cdots\times \{1\}_{A_N}\in M.
\end{aligned}
\end{equation}

\noindent {\bf Step 4:} Finally, we will prove that $M=\bbZ_3^{N-1}$.  For any $w\in \{0,2\}^{N-1}:=\{0,2\}_{A_2}\times\cdots\times \{0,2\}_{A_N}$.  We define $w_k\in \bbZ_3^{N-1}$ for $1\leq k\leq N$ as follows,
		\begin{equation}
		w_k^{(A_i)} = \left\{\begin{aligned}
		&\{1\}_{A_i} &\ifel i\leq k,\\
		&\{w\}_{A_i} &\ifel i > k.
		\end{aligned}\right.
		\end{equation}
		Note that $w_1=w$.
		Then we can prove $w_k\in M$ for $1\leq k\leq N$ by backward induction. The base case is $w_N = \{1\}_{A_2}\times \{1\}_{A_3}\times \cdots \times\{1\}_{A_N}\in M$ proved in Eq.~\eqref{eq:1111}. Next assume $w_k\in M$ for some $k\geq 2$ and consider $w_{k-1}$. Note that $w_k^{(A_k)}\neq w_{k-1}^{(A_k)}$, and $w_k^{(A_i)}= w_{k-1}^{(A_i)}$ for $i\neq k$. We can find two subcubes $\cC_{K_1}$ and $\cD_{K_2}$, correspondingly $\{0\}_{A_1}\times w_k\in\cC_{K_1}$ and $\{2\}_{A_1}\times w_k\in\cD_{K_2}$. Then by Step 3, as $w_k\in M$, we have $\cC_{K_1}^{(A_J)},\cD_{K_2}^{(A_J)}\subseteq M$.
  
		  By Table~\ref{Table:def}, we know that $\{\xi\}$ and $\{\eta\}$ appears alternatively in $\cC_{K_1}^{(A_2)}\times \cdots \times \cC_{K_1}^{(A_k)}$ and $\cD_{K_2}^{(A_2)}\times \cdots \times \cD_{K_2}^{(A_k)}$. Then either $\cC_{K_1}^{(A_k)}=\{\eta\}_{A_k}$, $\cD_{K_2}^{(A_k)}=\{\xi\}_{A_k}$ or $\cC_{K_1}^{(A_k)}=\{\xi\}_{A_k}$, $\cD_{K_2}^{(A_k)}=\{\eta\}_{A_k}$. Without loss of generality, we can assume that $\cC_{K_1}^{(A_k)}=\{\eta\}_{A_k}$, $\cD_{K_2}^{(A_k)}=\{\xi\}_{A_k}$. If  $w_{k-1}^{(A_k)}=\{0\}_{A_k}$, then $w_{k-1}\in \cC_{K_1}^{(A_J)}$; if $w_{k-1}^{(A_k)}=\{2\}_{A_k}$, then $w_{k-1}\in \cD_{K_2}^{(A_J)}$. In either case, we must have $w_{k-1}\in M$. Thus, by induction, $w_k\in M$ for $1\leq k\leq N$. It means that  $\{0,2\}^{N-1}\subseteq M$.
		
		For any $v\in\bbZ_3^{N-1}$, we can find a $\cP_K$, such that $v\in \cP_K^{(A_J)}$. Moreover, we can always find a $w\in \{0,2\}^{N-1}$, such that $w\in \cP_K^{(A_J)}$ by Lemma~\ref{lem:decomposition}. By Step 3, we know that $v\in M$. Thus $M=\bbZ_3^{N-1}$. This completes the proof.	
	\end{proof}
	
	\section{The proof of Theorem~\ref{thm:upb}}
	\label{appendix:upb}
		\begin{proof}
            {\bf Step 1: Transform the problem into the cube partition framework.}

		By Lemma~\ref{Lem:orthogonal_basis}, we know that $\{\ket{\cB_0}\}\cup\bigcup_{\begin{subarray}c
			K\in\indexset\\\cP\in\{\cC,\cD\}
			\end{subarray}}\ket{\cP_K}$ is a complete OPB. Since $|\indexset|=2^{N-1}$, it implies that there are $2^N$ product states of $\ket{\cP_K^+}$. Then $|\upb|=3^N-1-2^N+1=3^N-2^N$.
		
		Let $\cH$ be the space spanned by the states in $\upb$. For any state $\ket{\psi}\in\cH^{\bot}$, we only need to show that $\ket{\psi}$ must be an entangled state. We will prove it by contradiction. Assume that there exists a product state $\ket{\psi}\neq 0\in\cH^{\bot}$. Then there exists a unique subcube  $E = E_{A_1}\times E_{A_2}\times\dots\times E_{A_N}$ of $\bbZ_3^N$ with $|E_{A_i}|=e_i$ for $1\leq i\leq N$, and nonzero complex numbers $a_{j_i}^{(A_i)}$ for $j_i\in E_{A_i}$, such that
		\begin{equation}
		\label{eq:psi1}
		\ket{\psi}=\left(\sum_{j_1\in E_{A_1}}a^{(A_1)}_{j_1}\ket{j_1}\right)_{A_1}\otimes \left(\sum_{j_2\in E_{A_2}}a^{(A_2)}_{j_2}\ket{j_2}\right)_{A_2}\otimes\cdots\otimes\left(\sum_{j_N\in E_{A_N}}a^{(A_N)}_{j_N}\ket{j_N}\right)_{A_N}.
		\end{equation}
		In other words, $E$ is the support set of $\ket{\psi}$.
		Let $\cH_1$ be the space spanned the states in $\upb\setminus\{\ket{S}\}$, then $\cH_1^{\bot}$ is spanned by the states in  $\{\ket{\cB_0}\}\cup \bigcup_{\begin{subarray}c
			K\in\indexset\\\cP\in\{\cC,\cD\}
			\end{subarray}}\ket{\cP_K^+}$. Since $\cH_1\subseteq \cH$, it implies $\cH^{\bot}\subseteq \cH_1^{\bot}$. It means that $\ket{\psi}\in \cH_1^{\bot}$. Then there exist $b_{\cB_0}\in \bbC$ and $b_{\cP_K}\in\bbC$, such that
		\begin{equation}
		\label{eq:psi2}
		\ket{\psi}= b_{\cB_0}\ket{\cB_0} + \sum_{\begin{subarray}c
			K\in\indexset\\\cP\in\{\cC,\cD\}
			\end{subarray}} b_{\cP_K} \ket{\cP_K^+},
		\end{equation}
		and
		\begin{equation}
		\braket{\psi}{S}=0.
		\end{equation}
		
		By Comparing Eq.~\eqref{eq:psi1} and Eq.~\eqref{eq:psi2}, we know that
		\begin{equation}
		E = \left\{\begin{aligned}
		&\bigcup_{b_{\cP_K}\neq 0} \cP_K & \textrm{ if } b_{\cB_0}=0,\\
		&\cB_0\cup \bigcup_{b_{\cP_K}\neq 0} \cP_K  & \textrm{ if } b_{\cB_0}\neq 0.
		\end{aligned}
		\right.
		\end{equation}

            \noindent{\bf Step 2: $b_{\cP_K}$ is nonzero for at least two entries.}

		Firstly, if $b_{\cP_K}=0$ for all $\cP_K$, then $\ket{\psi}= b_{\cB_0}\ket{\cB_0}$ with $b_{\cB_0}\neq 0$, which contradicts to $\braket{\psi}{S}\neq 0$. Therefore, there is at least one nonzero $b_{\cP_K}$, i.e. at least one $\cP_K\subseteq E$. By Lemma~\ref{lem:decomposition}, $P_K$ contains one vector of $\{0,2\}^N$, so we must have
		\begin{equation}
		\label{eq:E_i}
		E_{A_i}\in \left\{\{0\}_{A_i},\{2\}_{A_i},\{\eta\}_{A_i},\{\xi\}_{A_i}, \{0,2\}_{A_i},\{0,1,2\}_{A_i}\right\}.
		\end{equation}

            Secondly, assume that there is only one nonzero $b_{\cP_K}$, then there are two cases. If $b_{\cB_0}=0$, then $\ket{\psi}=b_{\cP_K}\ket{\cP_K^+}$ for some $\cP_K$, which is not orthogonal to $\ket{S}$. Otherwise, $b_{\cB_0}\neq 0$, then $E=\cP_K\cup \cB_0$ is not a subcube, which is also impossible.
            \ \\

            \noindent{\bf Step 3:  Prove that $E$ contains the vectors in $\{0,2\}^n$}
		
		Now we are proving that $\{0,2\}_{A_i} \subseteq E_{A_i}$ for $1\leq i\leq N$ by contradiction. Assume that there exists $i_1$ such that $\{0,2\}_{A_{i_1}}\not\subseteq E_{A_{i_1}}$.  
        By Step 2 and Lemma~\ref{lem:decomposition},  there are two different vectors $v\neq w\in \{0,2\}^N\in E$. Suppose $v^{(A_{i_2})}\neq w^{(A_{i_2})}$, then $\{0,2\}_{A_{i_2}}\subseteq E_{A_{i_2}}$. (So apparently $i_2\neq i_1$.) By considering the cyclic property Eq.~\eqref{eq:cyclic}, we can always find $1\leq i\leq N$ such that
		\begin{equation}
		\{0,2\}_{A_i}\subseteq E_{A_i}\textrm{\; and \;} \{0,2\}_{A_{i+1}}\not\subseteq E_{A_{i+1}}.
		\end{equation}
		Then by Eq.~\eqref{eq:E_i}, $E_{A_{i+1}}\in \left\{\{0\}_{A_{i+1}},\{2\}_{A_{i+1}},\{\eta\}_{A_{i+1}},\{\xi\}_{A_{i+1}}\right\}$. There are the following two cases.
		\begin{enumerate}[label={(\alph*)}]
			\item Assume $E_{A_{i+1}} = \{0\}_{A_{i+1}}$ ($E_{A_{i+1}}=\{2\}_{A_{i+1}}$ is similar), then for all $\cP_K\subseteq E$, the only possible choice is 
			\begin{equation*}
			\cP_K^{(A_{i+1})}=\{0\}_{A_{i+1}}.
			\end{equation*}
			By Table~\ref{Table:def},  we know that $\cP_K^{(A_i)} \in \left\{\{0\}_{A_i}, \{\eta\}_{A_i}\right\}$. Therefore,
			\begin{equation}
			2\notin \bigcup_{b_{\cP_K}\neq 0}\cP_K^{(A_i)} \;\Rightarrow\; 2 \notin E_{A_i}.
			\end{equation}
			A contradiction.
			
			\item Assume $E_{A_{i+1}}=\{\eta\}_{A_{i+1}}$  ($E_{A_{i+1}}=\{\xi\}_{A_{i+1}}$ is similar).  Note that there exists a $\cP_K\subseteq E$ such that $\cP_K^{(A_{i+1})}=\{\eta\}_{A_{i+1}}$, because all possible values of $\cP_K^{(A_{i+1})}$ are just $\{0\}_{A_{i+1}},\{2\}_{A_{i+1}},\{\eta\}_{A_{i+1}},\{\xi\}_{A_{i+1}}$.
			
			Then by Table~\ref{Table:def} there are four possibilities of $\cP_K^{(A_i)}\times \cP_K^{(A_{i+1})}\times\cP_K^{(A_{i+2})}$, namely
			\begin{equation}
			\begin{aligned}
			\cF_1:=\{2\}_{A_i}\times\{\eta\}_{A_{i+1}}\times\{0\}_{A_{i+2}},\\
			\cF_2:=\{2\}_{A_i}\times\{\eta\}_{A_{i+1}}\times\{\xi\}_{A_{i+2}},\\
			\cF_3:=\{\xi\}_{A_i}\times\{\eta\}_{A_{i+1}}\times\{0\}_{A_{i+2}},\\
			\cF_4:=\{\xi\}_{A_i}\times\{\eta\}_{A_{i+1}}\times\{\xi\}_{A_{i+2}}.
			\end{aligned}
			\end{equation}
			Then we define another subcube $\cP_{K'}'$	by preserving $\cP_{K'}'^{(A_j)}=\cP_{K}^{(A_j)}$ for $j\notin\{i,i+1,i+2\}$, and
   \begin{small}
			\begin{equation}\label{eq:def}
			\cP_{K'}'^{(A_i)}\times \cP_{K'}'^{(A_{i+1})}\times\cP_{K'}'^{(A_{i+2})}\!= \!\left\{
			\begin{aligned}
			&\{\eta\}_{A_i}\times\{\xi\}_{A_{i+1}}\times\{\eta\}_{A_{i+2}} \ifel \cP_K^{(A_i)}\times \cP_K^{(A_{i+1})}\times\cP_K^{(A_{i+2})} =\cF_1,\\
			&\{\eta\}_{A_i}\times\{\xi\}_{A_{i+1}}\times\{2\}_{A_{i+2}}\ifel \cP_K^{(A_i)}\times \cP_K^{(A_{i+1})}\times\cP_K^{(A_{i+2})} =\cF_2,\\
			&\{0\}_{A_i}\times\{\xi\}_{A_{i+1}}\times\{\eta\}_{A_{i+2}} \ifel \cP_K^{(A_i)}\times \cP_K^{(A_{i+1})}\times\cP_K^{(A_{i+2})} =\cF_3,\\
			&\{0\}_{A_i}\times\{\xi\}_{A_{i+1}}\times\{2\}_{A_{i+2}} \ifel \cP_K^{(A_i)}\times \cP_K^{(A_{i+1})}\times\cP_K^{(A_{i+2})} = \cF_4.\\
			\end{aligned}
			\right.
			\end{equation}
    \end{small}
			Then we can check that $|K'|$ is still even and $\cP_{K'}'$ satisfies Table~\ref{Table:def}, so $\cP_{K'}'\in \{\cC_K,\cD_K\}_{K\in\indexset}$. Moreover, we consider an vector $v\in \cP_K$ with $v^{(A_{i})}=\{2\}_{A_i}$, $v^{(A_{i+1})}=\{1\}_{A_{i+1}}$ and $v^{(A_{i+2})}\in\{\{0\}_{A_{i+2}},\{2\}_{A_{i+2}}\}$.
			Next, define $v'$, where $v'^{(A_k)}=v^{(A_k)}$ for $k\neq i$, and $v'^{(A_i)}=\{0\}_{A_i}$.
			Then by definition of $\cP_{K'}'$,  we know that $v'\in \cP_{K'}'$. However, as $\{0,2\}_{A_i}\subseteq E_{A_i}$, we also have $v'\in E$, which means $P_{K'}'\subseteq E$ and $P_{K'}'^{(A_{i+1})}=\{\xi\}_{A_{i+1}} \subseteq E_{A_{i+1}}$. This contradicts to $E_{A_{i+1}}=\{\eta\}_{A_{i+1}}$.
		\end{enumerate}
		
		Thus we proved $\{0,2\}^N\subseteq E$ by contradiction.
            \ \\

            \noindent{\bf Step 4: Show $E$ is the whole cube $\bbZ_3^n$ and all $b_{\cP_K}$ are the same}
    
            By Lemma~\ref{lem:decomposition},
		\begin{equation}
		\{0,2\}^N\subseteq E \quad \Longrightarrow \quad \cP_K\subseteq E \textrm{ for all } \cP_K \quad \Longrightarrow \quad \bbZ_3^N \setminus \cB_0 \subseteq E.
		\end{equation}
		And since $E$ is a subcube, we must have
		$E = \bbZ_3^N.$\\
  
            Then there are nonzero $a^{(A_i)}_{j_i}$ for $j_i\in \bbZ_3$ and $1\leq  i\leq N$, such that
		\begin{equation}\label{eq:coeff}
		\ket{\psi}=\left(\sum_{j_1\in \bbZ_3}a^{(A_1)}_{j_1}\ket{j_1}\right)_{A_1}\otimes\left(\sum_{j_2\in \bbZ_3}a^{(A_2)}_{j_2}\ket{j_2}\right)_{A_2}\otimes\cdots\otimes\left(\sum_{j_N\in \bbZ_3}a^{(A_N)}_{j_N}\ket{j_N}\right)_{A_N}.
		\end{equation}
            To make a contradiction to $\braket{\psi}{S} = 0$, we will show that $\ket{\psi}\propto \ket{S}$. 
  
		To simplify the discussion, for all $1\leq i\leq N$, we define
		\begin{equation}
		\begin{aligned}
		\ket{v_i} &:= \ket{0}_{A_1}\otimes\dots \otimes\ket{0}_{A_{i-1}}\otimes\ket{1}_{A_i}\otimes\ket{0}_{A_{i+1}}\otimes\dots \otimes\ket{0}_{A_N},\\
		\ket{v_i'} &:= \ket{0}_{A_1}\otimes\dots \otimes\ket{0}_{A_{i-1}}\otimes\ket{2}_{A_i}\otimes\ket{0}_{A_{i+1}}\otimes\dots \otimes\ket{0}_{A_N},\\
		\cR_i &:= \{0\}_{A_1}\times\{0\}_{A_2}\times\dots \times\{\xi\}_{A_i}\times\{\eta\}_{A_{i+1}}\times \dots\times\{0\}_{A_N}.
		\end{aligned}
		\end{equation}
		Then $\cR_i \in \{\cC_K,\cD_K\}_{K\in\indexset}$ containing the vectors $v_i,v_i'$. By computing the coefficients of Eq.~\eqref{eq:coeff}, for $1\leq i\leq N$, we have
		\begin{equation}
		\begin{aligned}
		\prod_{j\neq i}a_0^{(A_j)}\cdot a_1^{(A_i)}=\text{Coeff. of $\ket{v_i}$} =b_{\cR_i}= \text{Coeff. of $\ket{v_i'}$} = \prod_{j\neq i}a_0^{(A_j)}\cdot a_2^{(A_i)}
		\end{aligned}
		\end{equation}
		Then it implies that $a_1^{(A_i)}=a_2^{(A_i)}$ for $1\leq i\leq N$. Similarly, by using
		\begin{equation}
		\begin{aligned}
		&\ket{2}_{A_1}\otimes\cdots\otimes \ket{2}_{A_{i-1}}\otimes \ket{0}_{A_i}\otimes\ket{2}_{A_{i+1}}\otimes\cdots \otimes\ket{2}_{A_N},\\
		&\ket{2}_{A_1}\otimes\cdots\otimes \ket{2}_{A_{i-1}}\otimes \ket{1}_{A_i}\otimes\ket{2}_{A_{i+1}}\otimes\cdots\otimes \ket{2}_{A_N},\\
		&\{2\}_{A_1}\times\{2\}_{A_2}\times\dots \times\{\eta\}_{A_i}\times\{\xi\}_{A_{i+1}} \times\dots\times\{2\}_{A_N},
		\end{aligned}
		\end{equation}
		we can also obtain that $a_1^{(A_i)}=a_0^{(A_i)}$ for $1\leq i\leq N$.  Therefore, we conclude that
		\begin{equation}
		\ket{\psi} = \bigotimes_{i=1}^N \left(a_0^{(A_i)}\left(\ket{0}+\ket{1}+\ket{2}\right)_{A_i}\right) = \prod_{i=1}^N a_0^{(A_i)}\cdot \ket{S},
		\end{equation}
		which is not orthogonal to $\ket{S}$. This completes the proof.
	\end{proof}

\bibliographystyle{unsrtnat}
\bibliography{reference}	
\end{document}